\documentclass[border=2px,11pt]{article}

\usepackage{indentfirst}
\usepackage[scale=0.8]{geometry}
\usepackage{authblk}

\usepackage{amsmath,amssymb,amsthm,amsfonts,mathdots}
\usepackage{latexsym,bm}
\usepackage{color}
\usepackage{arydshln}

\usepackage{graphicx}
\usepackage{float}

\usepackage{tikz}
\usetikzlibrary{decorations.pathreplacing}
\newcommand{\tikzmark}[1]{\tikz[overlay,remember picture,baseline=(#1.base)]
  \node (#1) {\strut};}

\usepackage[ruled]{algorithm2e}

\usepackage{hyperref}

\newtheorem{theorem}{Theorem}[section]
\newtheorem{lemma}{Lemma}[section]

\newtheorem{definition}{Definition}[section]
\newtheorem{example}{\bf Example}[section]
\theoremstyle{remark}
\newtheorem{remark}{\bf Remark}[section]
\numberwithin{equation}{section}

\def \d {{\rm{d}}}
\def \e {{\rm{e}}}
\def \i {{\rm{i}}}

\newcommand{\ket}[1]{| #1 \rangle} 
\newcommand{\bra}[1]{\langle #1 |} 

\newcommand{\bb}{\boldsymbol}
\DeclareMathOperator{\polylog}{polylog}

\author[1,2,3]{Shi Jin\thanks{shijin-m@sjtu.edu.cn}}
\author[1]{Shuyi Zhang\thanks{shuyi-zhang@sjtu.edu.cn}}
\affil[1]{\small School of Mathematical Sciences, Shanghai Jiao Tong University, Shanghai, 200240, China}
\affil[2]{\small Institute of Natural Sciences, Shanghai Jiao Tong University, Shanghai 200240, China}
\affil[3]{\small Ministry of Education Key Laboratory in Scientific and Engineering Computing, Shanghai Jiao Tong University, Shanghai 200240, China}

\date{}

\begin{document}

\title{Quantum simulation of Liouville equation in geometrical optics with partial transmission and reflection via Schr\"odingerization}

\maketitle

\begin{abstract}
    This paper investigates quantum simulation algorithms for the Liouville equation in geometrical optics with partial transmission and reflection at sharp interfaces, based on the  Schr\"odingerization method. By means of a warped phase transformation in one higher dimension, the Schr\"odingerization method converts any linear PDEs into a system of Schr\"odinger-type equations with unitary evolution\cite{jinQuantumSimulationPartial2023,jinQuantumSimulationPartial2024a}, 
    thereby making them suitable for quantum simulation. 
    The Schr\"odingerization method is combined  with the Hamiltonian-preserving scheme developed in \cite{jinHamiltonianPreservingSchemeLiouville2006}, which builds into the numerical flux partial transmissions and reflections. We overcome the difficulty of quantum simulation when facing the ``if/else'' judgment at the interface by encoding a partial transmission and reflection matrix {\it a priori}, not during the time evolution.   
    Detailed constructions of the quantum algorithms are presented, together with complexity analysis showing that the proposed methods achieve polynomial quantum advantage in precision $\epsilon$ over their classical counterparts.
\end{abstract}



\section{Introduction}

Geometrical optics, as an approximation of wave optics in the high-frequency regimes, serves as a powerful tool for analyzing and designing optical systems. While traditional ray tracing is intuitive \cite{glassnerIntroductionRayTracing1989}, it often faces computational bottlenecks when handling massive ray volumes to statistically characterize light field distributions (e.g., luminance, irradiance). In contrast, the Liouville equation, which describes the distribution of light energy in phase space, provides a more universal and continuous theoretical framework \cite{kurtbernardowolfGeometricOpticsPhase2004}. Originating from classical mechanics, this equation reveals the evolution of particle distribution functions in Hamiltonian systems \cite{hamiltonTheorySystemsRays1828}. In geometrical optics, it governs the transport of light energy within the position-momentum phase space, serving as a bridge between ray optics and wave optics.

The Liouville equation in geometrical optics can be expressed as a first-order linear hyperbolic partial differential equation \cite{engquistComputationalHighFrequency2003,jinAsymptoticpreservingSchemesMultiscale2022}, formulated as 
\begin{equation}\label{equ:ddimLiouville}
    f_t+H_\mathbf{v}\cdot\nabla_\mathbf{x}f-H_\mathbf{x}\cdot\nabla_\mathbf{v}f=0,\quad t>0,\quad\mathbf{x},\mathbf{v}\in R^d
\end{equation} where the Hamiltonian is given by 
\begin{equation}\label{equ:LiouvilleHamilOptics}
    H(t,\mathbf{x},\mathbf{v})=c(\mathbf{x})|\mathbf{v}|=c(\mathbf{x})\sqrt{v_1^2+v_2^2+\cdots+v_d^2}
\end{equation}
with $c(\mathbf{x}) > 0$ being the local wave speed, or the reciprocal of the index of refraction, $f(t, \mathbf{x}, \mathbf{v})$ is the density distribution of particles depending on time $t$, position $\mathbf{x}$ and the slowness vector $\mathbf{v}$. In this paper, we are interested in the case when $c(\mathbf{x})$ contains discontinuities corresponding to different indices of refraction at different media. This discontinuity will generate an interface at the point of discontinuity of $c(\mathbf{x})$, and as a consequence waves crossing this interface will undergo transmissions and reflections. The incident and transmitted waves obey Snell's law of refraction.

The Liouville framework has found widespread applications in various scientific and engineering disciplines. In optics, it is employed to model light propagation through complex refractive media, lens systems, and waveguides \cite{kurtbernardowolfGeometricOpticsPhase2004}. In plasma physics, it describes electromagnetic wave transport in inhomogeneous plasmas \cite{ryzhikTransportEquationsElastic1996,engquistComputationalHighFrequency2003}. In geophysics and atmospheric science, analogous equations are used for the propagation of seismic waves, acoustic waves, and radiative transfer \cite{ryzhikTransportEquationsElastic1996,engquistComputationalHighFrequency2003}. Moreover, the Liouville equation serves as a bridge between classical and quantum descriptions, closely related to the Wigner transform and semiclassical limits of the Schr\"odinger equation \cite{jinAsymptoticpreservingSchemesMultiscale2022,engquistComputationalHighFrequency2003}. However, numerical solutions to this equation face significant challenges: First, the high dimensionality of phase space (six dimensions for 3D problems) subjects grid-based numerical methods to the Curse of Dimensionality, where computational cost grows exponentially with dimension. Second, solutions may contain discontinuities (e.g., interface problems), demanding higher precision and stability for numerical approximations.

To address these challenges,  various numerical strategies have been developed in the literature.  These can be broadly categorized into two types:
\begin{itemize}
    \item Eulerian Methods: Discretize equations on a fixed grid, such as upwind finite difference methods \cite{vanlithNovelSchemeLiouvilles2016} and the Discontinuous Galerkin Method \cite{roberta.m.vangestelADERDiscontinuousGalerkin2023}. These methods allow accurate computations but incur high computational costs in high dimensions.
    \item Lagrangian/Particle Methods: Examples include wavefront methods \cite{engquistComputationalHighFrequency2003} and symplectic integrators \cite{sanz-sernaRungekuttaSchemesHamiltonian1988,kangHamiltonianWayComputing1991,leimkuhlerSimulatingHamiltonianDynamics2005} for many samples. These approaches represent solutions as collections of moving particles, making them inherently suited for high-dimensional problems. However, when particles are far away from each others, one has to add more particles in order to guarantee sufficient numerical accuracy. This is difficult to implement numerically.
\end{itemize}
In \cite{engquistHighFrequencyWavePropagation2002,fomelFastphaseSpaceComputation2002,jinComputingMultivaluedPhysical2005a,osherGeometricOpticsPhaseSpaceBased2002}, several phase space based level set methods, which give rise to   Liouville equations, are developed. High frequency limit of wave equations with transmissions and reflections at the interfaces was studied in \cite{guillaumebalTransportTheoryAcoustic1999,millerRefractionHighfrequencyWaves2000,ryzhikTransportEquationsWaves1997}. A Liouville equation based level set method for the wave front, but with only reflection, was introduced in \cite{chengReflectionLevelSet2004}. It was also suggested to smooth out the local wave speed in \cite{osherGeometricOpticsPhaseSpaceBased2002}. When $c(\mathbf{x})$ is smooth, standard numerical methods for linear wave equations give satisfactory results. However, if $c(\mathbf{x})$ is discontinuous, the conventional numerical schemes suffer from the following problems: 
\begin{enumerate}
    \item with $\Delta x$ the mesh size in the physical space and $\Delta v$ the mesh size in particle slowness space, an explicit scheme needs time step $\Delta t=\mathcal{O}(\Delta x\Delta v)$ which is too expensive;
    \item Hamiltonian is not preserved across the interface, usually leading to poor or even incorrect numerical resolutions by ignoring the discontinuities of $c(\mathbf{x})$;
\end{enumerate}   
The Hamiltonian-preserving schemes in \cite{jinHamiltonianPreservingSchemesLiouville2005,jinHamiltonianPreservingSchemeLiouville2006,jinHamiltonianpreservingSchemesLiouville2006} overcame the above  problems. They preserve the Hamiltonian across the interface, give a selection criterion for a unique solution to the governing equation and allow a typical hyperbolic stability condition $\Delta t= O(\Delta x,\Delta v)$. In particular, by building the transmission and reflection mechanism---or the Snell's Law---into the numerical flux, they automatically capture the correct physical interface behavior. 

Although classical algorithms have advanced rapidly and can now solve many problems, the development of classical computers faces physical limitations \cite{freiserSurveyPhysicalLimitations1969}, such as the failure of Moore's Law. At the same time, people still aspire to consume fewer resources and compute faster. With the advancement of quantum computer, this vision may become a reality. The idea of quantum computer was first proposed by Richard Feynman \cite{feynmanSimulatingPhysicsComputers1982}. Later several quantum algorithms \cite{shorAlgorithmsQuantumComputation1994,daviddeutschQuantumTheoryChurch1985,harrowQuantumAlgorithmLinear2009} theoretically proves the potential quantum acceleration on specific problems. Among them, the HHL algorithm for solving a linear system of algebraic equations \cite{harrowQuantumAlgorithmLinear2009} shows possible exponential speedup over its classical counterparts.  Then many variants of HHL algorithm \cite{ambainisVariableTimeAmplitude2012,childsQuantumAlgorithmSystems2017,subasiQuantumAlgorithmsSystems2019} were proposed. By discretizing PDEs into linear systems, one can solve PDEs by QLSA \cite{berryHighorderQuantumAlgorithm2014,berryQuantumAlgorithmLinear2017,childsQuantumSpectralMethods2020}. Another way to solve PDEs on quantum computer is by Hamiltonian simulation, which is first proposed in \cite{lloyd1996universal} and used in \cite{costaQuantumAlgorithmSimulating2019} for wave equations.

In recent years, a novel universal method known as Schr\"odingerization has been proposed \cite{jinQuantumSimulationPartial2024a,jinQuantumSimulationPartial2023}. It first employs a warped phase transformation, which, by adding an extra dimension to the equation, converts any linear PDE with non-unitary dynamics into a Schr\"odinger-like system of equations with unitary dynamics, enabling subsequent Hamiltonian simulation. 
Related to this study, it has been applied to equations with physical boundary or interface conditions \cite{jinQuantumSimulationPartial2024}, and quantum dynamics with artificial boundary conditions \cite{jinQuantumSimulationQuantum2024}, among many other problems. Error estimates of Schr\"odingerization and the important issue of recovering the original variables from the Schr\"odingerized equations are studied in \cite{jinSchrodingerizationBasedQuantumAlgorithms2025}. By choosing smoother initial functions in auxiliary space, Schr\"odingerization can in fact achieve near optimal and even optimal scaling in matrix queries \cite{jinSchrodingerizationMethodLinear2025a}. Apart from implementation based on qubits, Schr\"odingerization method can be also extended to the qumode framework \cite{jinAnalogQuantumSimulation2024a,jinAnalogQuantumSimulation2024}. 

In this work, we investigate the use of Schr\"odingerization method for Liouville equation in geometrical optics with interface. We apply the Schr\"odingerization approach combined with the Hamiltonian-preserving scheme in \cite{jinHamiltonianPreservingSchemeLiouville2006} to one and two space dimensional cases. 
A major difficulty arises from the fact that the Hamiltonian-preserving scheme treats partial transmission and reflection at the interface using a standard procedure based on ``if/else'' statements to check threshold conditions. This feature makes it highly nontrivial to cast the scheme into a matrix form suitable for quantum simulation. To address this issue, we encode the corresponding interface conditions into a matrix representation through certain nonlinear functions, which is prepared {\it a priori}, not during the time evolution. On this basis, we carry out a query complexity analysis and prove that the proposed methods achieve a polynomial quantum advantage in precision $\epsilon$ over their classical counterparts.
We then give the implementation details of the application of the Schr\"odingerization method to some interface problems for Liouville equation in geometrical optics.

This paper is organized as follows. In section \ref{sec:introLiouville}, we review the derivation of the Liouville equation and two interface conditions. Section \ref{sec:Schr} reviews the Schr\"odingerization method and gives an example which will be used in later sections. In section \ref{sec:interface1}, we consider Hamiltonian-preserving scheme for the Liouville equations \cite{jinHamiltonianPreservingSchemeLiouville2006} in one dimensional and convert it to quantum algorithms using Schr\"odingerization. We extend the quantum algorithms to the two space dimension in section \ref{sec:interface2d} in the simple case of an interface aligning with the grids and analyse the complexity. Numerical examples are given in section \ref{sec:Numerical} to verify the effectivity of Schr\"odingerization methods.
We make some concluding remarks in section \ref{sec:conclusion}.

\section{Liouville equation in geometrical optics}\label{sec:introLiouville}

The Liouville equation arises in the phase space description of geometrical optics. It is  the high frequency limit of the wave equation
\begin{equation}
    u_{tt}-c(\mathbf{x})^2\Delta u=0,\quad t>0,\quad\mathbf{x}\in R^d,
\end{equation} 
where $c(\mathbf{x})>0$ is the local speed of wave propagation of the medium, or the reciprocal of the index of refraction. If the initial data contain high frequency oscillations, it leads to highly oscillatory solutions to the equation, making it computationally prohibitive \cite{engquistComputationalHighFrequency2003}.

Geometrical optics method is often used to get around this computational burden. It works by inserting the ansatz
\begin{equation}
    u\approx e^{i\omega\tilde{\phi}(\mathbf{x},t)}\sum_{j=0}^\infty A_j(\mathbf{x},t)(i\omega)^{-j}
\end{equation}
into the wave equation, which  leads to 
\begin{equation}
    \omega^2\left(\left(\frac{\partial\tilde{\phi}}{\partial t}\right)^2-c(\mathbf{x})^2|\nabla\tilde{\phi}|^2\right)=0
\end{equation} as the leading coefficient term. Thus one gets the eikonal equation
\begin{equation}
    \frac{\partial\tilde{\phi}}{\partial t}\pm c(\mathbf{x})|\nabla\tilde{\phi}|=0.
\end{equation}
Defining $\nabla\tilde{\phi}=(v_1,v_2,\ldots,v_n)=\mathbf{v}$, its characteristics are given by the system of ordinary differential equations \cite{osherGeometricOpticsPhaseSpaceBased2002}
\begin{equation}\label{equ:ddimLiouvillebicha}
    \frac{d\mathbf{x}}{dt}=c(\mathbf{x})\frac{\mathbf{v}}{|\mathbf{v}|},\quad\frac{d\mathbf{v}}{dt}=-|\mathbf{v}|\nabla c(\mathbf{x}).
\end{equation}
These equations give the velocity field under which the wavefronts move. It is also possible to derive from this a PDE that transports values of a function along these characteristic directions. Let $f(\mathbf{x},\mathbf{v}, t)\ge 0$ be  the density distribution of particles (rays). The PDE for $f$ that transports its values along the characteristic directions then takes the form of a Liouville equation:
\begin{equation}\label{equ:ddimLiouvilleOptics}
    f_t+c(\mathbf{x})\frac{\mathbf{v}}{|\mathbf{v}|}\cdot\nabla_{\mathbf{x}}f-|\mathbf{v}|\nabla_{\mathbf{x}}c(\mathbf{x})\cdot\nabla_{\mathbf{v}} f=0.
\end{equation}
Solving this equation yields a global distribution of the light field, eliminating the randomness inherent in traditional ray tracing. It provides a powerful mathematical model for Eulerian computational methods that allow accurately predicting and optimizing the performance of optical systems.

\subsection{The behavior of waves at an interface}

This equation essentially describes a conservation law in phase space: the optical energy density distribution remains constant along the ray trajectory. The bicharacteristics of this Liouville equation \eqref{equ:ddimLiouvilleOptics} (i.e., the ray trajectory), satisfy the Hamiltonian system \eqref{equ:ddimLiouvillebicha}. When there is an interface where the waves can partially transmit and reflect, the Liouville equation as well as its bi-characteristics \eqref{equ:ddimLiouvillebicha}, is not well-defined. 
On the other hand, in classical mechanics, the Hamiltonian $H$ of a particle remains constant along the trajectory of the particle, even when it is being transmitted or reflected by the interface, i.e., the Hamiltonian should be preserved across the interface
\begin{equation}\label{equ:ddimLiouvillePreserve}
    H(t,\mathbf{x}^+,\mathbf{v}^+)=H(t,\mathbf{x}^-,\mathbf{v}^-).
\end{equation}
Specifically, when $c(\mathbf{x})$ contains discontinuities corresponding to different indices of refraction at different media, this discontinuity will generate an interface at the point of discontinuity of $c(\mathbf{x})$ and \eqref{equ:ddimLiouvillePreserve} becomes \begin{equation}\label{equ:ddimLiouvilleOpticsPreserve}
    c^+|\mathbf{v}^+|=c^-|\mathbf{v}^-|.
\end{equation}
where the superscripts $\pm$ indicate the right and left limits of the quantity at the interface.
As a consequence waves crossing this interface will undergo partial transmissions and reflections.
When a plane wave hits a flat interface, \eqref{equ:ddimLiouvilleOpticsPreserve} is equivalent to Snell's law of refraction
\begin{equation}
    \frac{\sin\theta_i}{c^-}=\frac{\sin\theta_t}{c^+},
\end{equation}
and the reflection law
\begin{equation}
    \theta_r=\theta_i,
\end{equation}
where $\theta_i$, $\theta_t$, and $\theta_r$ stand for angles of incident and transmitted and reflected waves, as shown in Figure \ref{fig:transmissionreflection}.
The reflection coefficient is given by
\begin{equation}\label{equ:coefficient}
    \alpha^{R}=\left(\frac{c^{+}\cos\theta_{i}-c^{-}\cos\theta_{t}}{c^{+}\cos\theta_{i}+c^{-}\cos\theta_{t}}\right)^{2}
\end{equation}
while the transmission coefficient is $\alpha^T=1-\alpha^R$.

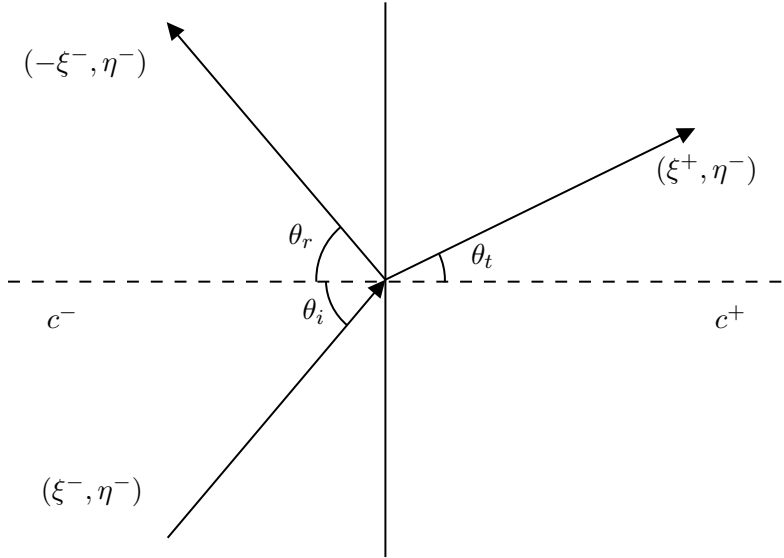
\begin{figure}[!htb]
    \centering
    
    \tikzset{every picture/.style={line width=0.75pt}} 

    \begin{tikzpicture}[x=0.75pt,y=0.75pt,yscale=-1,xscale=1]

    \draw  [dash pattern={on 4.5pt off 4.5pt}]  (140,191) -- (530.33,191.02) ;
    \draw    (330.33,50.02) -- (330.33,330.02) ;
    \draw    (220.5,320.28) -- (328.4,192.31) ;
    \draw [shift={(330.33,190.02)}, rotate = 130.14] [fill={rgb, 255:red, 0; green, 0; blue, 0 }  ][line width=0.08]  [draw opacity=0] (8.93,-4.29) -- (0,0) -- (8.93,4.29) -- cycle    ;
    \draw    (330.33,190.02) -- (221.94,62.27) ;
    \draw [shift={(220,59.98)}, rotate = 49.69] [fill={rgb, 255:red, 0; green, 0; blue, 0 }  ][line width=0.08]  [draw opacity=0] (8.93,-4.29) -- (0,0) -- (8.93,4.29) -- cycle    ;
    \draw    (330.33,190.02) -- (483.92,114.84) ;
    \draw [shift={(486.61,113.52)}, rotate = 153.92] [fill={rgb, 255:red, 0; green, 0; blue, 0 }  ][line width=0.08]  [draw opacity=0] (8.93,-4.29) -- (0,0) -- (8.93,4.29) -- cycle    ;
    \draw  [draw opacity=0] (311.23,213.14) .. controls (304.74,207.78) and (300.54,199.73) .. (300.34,190.7) -- (330.33,190.02) -- cycle ; \draw   (311.23,213.14) .. controls (304.74,207.78) and (300.54,199.73) .. (300.34,190.7) ;  
    \draw  [draw opacity=0] (357.46,177.18) .. controls (359.3,181.07) and (360.33,185.42) .. (360.33,190.02) .. controls (360.33,190.46) and (360.32,190.9) .. (360.3,191.35) -- (330.33,190.02) -- cycle ; \draw   (357.46,177.18) .. controls (359.3,181.07) and (360.33,185.42) .. (360.33,190.02) .. controls (360.33,190.46) and (360.32,190.9) .. (360.3,191.35) ;  
    \draw  [draw opacity=0] (295.46,191.05) .. controls (295.45,190.71) and (295.45,190.36) .. (295.45,190.02) .. controls (295.45,179.39) and (300.21,169.87) .. (307.71,163.47) -- (330.33,190.02) -- cycle ; \draw   (295.46,191.05) .. controls (295.45,190.71) and (295.45,190.36) .. (295.45,190.02) .. controls (295.45,179.39) and (300.21,169.87) .. (307.71,163.47) ;  

    \draw (280,160.4) node [anchor=north west][inner sep=0.75pt]    {$\theta _{r}$};
    \draw (287,198.4) node [anchor=north west][inner sep=0.75pt]    {$\theta _{i}$};
    \draw (372,170.4) node [anchor=north west][inner sep=0.75pt]  [font=\normalsize]  {$\theta _{t}$};
    \draw (146,71.4) node [anchor=north west][inner sep=0.75pt]    {$\left( -\xi ^{-} ,\eta ^{-}\right)$};
    \draw (155,287.4) node [anchor=north west][inner sep=0.75pt]    {$\left( \xi ^{-} ,\eta ^{-}\right)$};
    \draw (465,125.4) node [anchor=north west][inner sep=0.75pt]    {$\left( \xi ^{+} ,\eta ^{-}\right)$};
    \draw (158,201.4) node [anchor=north west][inner sep=0.75pt]    {$c^{-}$};
    \draw (495,201.4) node [anchor=north west][inner sep=0.75pt]    {$c^{+}$};

    \end{tikzpicture}
    \caption{Wave transmission and reflection at an interface.}\label{fig:transmissionreflection}
\end{figure}

We will discuss this behavior in more detail in 1D and 2D, respectively.
\begin{itemize}
    \item The 1D case is simpler. Consider the case when, at an interface, the characteristic on the left of the interface is given by $\xi^->0.$ Then with probability $\alpha^{R}=\left(\frac{c^{+}-c^{-}}{c^{+}+c^{-}}\right)^{2}$, the wave is reflected by the interface with a new velocity $-\xi^-$, and with probability $\alpha^T=1-\alpha^R$ it will cross the interface with the new velocity $\xi^+=\frac{c^-}{c^+}\xi^-$ determined by \eqref{equ:ddimLiouvilleOpticsPreserve}.
    \item The 2D case, when an incident wave hits a vertical interface (see Figure \ref{fig:transmissionreflection}). Let  $\boldsymbol{x}=(x,y)$, $\boldsymbol{v}=(\xi,\eta).$ Assume that the incident wave has a velocity $(\xi^-,\eta^-)$ to the left side of the interface, with $\xi^->0.$ Since the interface is vertical, \eqref{equ:ddimLiouvillebicha} implies that $\eta$ is not changed when the wave crosses the interface. There are two possibilities:
    \begin{enumerate}
        \item $\left(\frac{c^{-}}{c^{+}}\right)^{2}(\xi^{-})^{2}+\left[\left(\frac{c^{-}}{c^{+}}\right)^{2}-1\right](\eta^{-})^{2}>0.$ In this case the wave can partrially transmit and partially be reflected. With probability $\alpha^{R}=\left(\frac{c^{+}\gamma^{-}-c^{-}\gamma^{+}}{c^{+}\gamma^{-}+c^{-}\gamma^{+}}\right)^{2}$ the wave is reflected with a new velocity $(-\xi^{-},\eta^{-})$, where 
        \begin{equation*}
        \gamma^{+}=\cos(\theta_{t})=\frac{\xi^{+}}{\sqrt{\left(\xi^{+}\right)^{2}+\left(\eta^{-}\right)^{2}}},\quad\gamma^{-}=\cos(\theta_{i})=\frac{\xi^{-}}{\sqrt{\left(\xi^{-}\right)^{2}+\left(\eta^{-}\right)^{2}}}.
        \end{equation*}
        With probability $\alpha^T=1-\alpha^R$ it will be transmitted with the new velocit $(\xi^+,\eta^-)$, where 
        \begin{equation*}
        \xi^{+}=\sqrt{\left(\frac{c^{-}}{c^{+}}\right)^{2}(\xi^{-})^{2}+\left[\left(\frac{c^{-}}{c^{+}}\right)^{2}-1\right](\eta^{-})^{2}},
        \end{equation*}
        is obtained using \eqref{equ:ddimLiouvilleOpticsPreserve}.
        \item $c^{- }< c^{+ }$ and $\left ( \frac {c^{- }}{c^{+ }}\right ) ^{2}( \xi ^{- }) ^{2}+ \left [ \left ( \frac {c^{- }}{c^{+ }}\right ) ^{2}- 1\right ] ( \eta ^{- }) ^{2} < 0$. In this case, it is impossible for the wave to transmit, so the wave will be completely reflected with velocity $(-\xi^{-},\eta^{-}).$
    \end{enumerate}
\end{itemize}

If $\xi^-<0$, similar behavior can also be analyzed using the constant Hamiltonian condition \eqref{equ:ddimLiouvilleOpticsPreserve}.

When the wave length of the incident wave is much shorter than the width of the interface as both lengths go to zero, we can only consider the transmission at the interface. To build in the Hamiltonian-preserving mechanism into the numerical schemes,  one can used the fact that the density distribution $f$ remains unchanged along the characteristic, thus 
\begin{equation}\label{equ:ddimLiouvilleOpticsinterfacecondition}
    f(t,\mathbf{x}^+,\mathbf{v}^+)=f(t,\mathbf{x}^-,\mathbf{v}^-)
\end{equation}
at a discontinuous point $\mathbf{x}$ of $c(\mathbf{x})$, where $\mathbf{v}^+$ and $\mathbf{v}^-$ satisfies the constant Hamiltonian condition \eqref{equ:ddimLiouvilleOpticsPreserve}. 

When the wave length is much longer than the width of the interface, while both lengths go to zero, the waves can be partially transmitted and reflected. With partial transmissions and reflections, $f$ doesn't remain a constant along a bicharacteristic, since $f$ needs to be determined from two bicharacteristics, one accounting for the transmission and the other for reflection. \cite{jinHamiltonianPreservingSchemeLiouville2006} uses the following condition at the interface:
\begin{equation}\label{equ:ddimLiouvilleOpticsinterfaceconditionPartial}
    f(t,\mathbf{x}^+,\mathbf{v}^+)=\alpha^Tf(t,\mathbf{x}^-,\mathbf{v}^-)+\alpha^Rf(t,\mathbf{x}^+,\mathbf{v}^+_R)
\end{equation}
where $\mathbf{v}^-$ is defined from $\mathbf{v}^+$ through the constant Hamiltonian condition \eqref{equ:ddimLiouvilleOpticsPreserve}, $\mathbf{v}^+_R$ is defined from $\mathbf{v}^+$ through the law of reflection, $\alpha^T$ and $\alpha^R$ are the transmission and reflection coefficients which add up to $1$ and vary with $\mathbf{v}^+$ except in the 1D case. When $\alpha^T=1$ and $\alpha^R=0$, the interface condition \eqref{equ:ddimLiouvilleOpticsinterfaceconditionPartial} degenerates to \eqref{equ:ddimLiouvilleOpticsinterfacecondition}, so one only needs to study the interface condition \eqref{equ:ddimLiouvilleOpticsinterfaceconditionPartial}.

\section{Overview of the Schr\"odingerization method} \label{sec:Schr}

\subsection{The Schr\"odingerization method}

Consider the following linear dynamical system:
\begin{equation}\label{equ:ODElinear}
 \begin{cases}
 \frac{\d \boldsymbol{u}(t)}{\d t} = A \boldsymbol{u}(t) + \boldsymbol{b}(t), \\
 \boldsymbol{u}(0) = \boldsymbol{u}_0,
 \end{cases}
\end{equation}
where $\boldsymbol{u}, \boldsymbol{b} \in \mathbb{C}^n$  and $A \in \mathbb{C}^{n\times n} $. In general, $A $ is non-Hermitian, namely,  $A^{\dagger} \neq A$, where ``$\dagger$'' denotes the conjugate transpose. For simplicity, We focus on the homogeneous case $\bb{b}(t)=0$. The inhomogeneous term $\bb{b}(t)\neq0$ can be handled by introducing an auxiliary constant vector function, so that it is absorbed into the matrix $A$ \cite{jinQuantumSimulationMaxwells2024,jinSchrodingerizationBasedQuantumAlgorithms2025,huQuantumAlgorithmsMultiscale2024}.

Any matrix $A$ can be decomposed into a Hermitian term and an anti-Hermitian term:
\[A = H_1 + \i H_2, \qquad \i = \sqrt{-1},\]
where
\begin{equation}\label{equ:H1H2}
    H_1 = \frac{A+A^{\dagger}}{2} = H_1^{\dagger}, \qquad H_2 = \frac{A-A^{\dagger}}{2 \i} = H_2^{\dagger}.
\end{equation}

A natural assumption is that the semi-discrete system \eqref{equ:ODElinear} inherits the stability of the original PDE, which implies that $H_1$ is negative semi-definite. That is, there exists a unitary matrix $Q$ such that $Q^{-1} H_1 Q = \text{diag}(\lambda_1, \cdots, \lambda_n)=:\Lambda$ with $\lambda_j\le 0$ for all $1\le j\le n$. If $H_1$ has positive eigenvalues, we can choose to shift all the eigenvalues \cite{guQuantumSimulationClass2025}.

Using the warped phase transformation 
\begin{equation}\label{equ:warped}
    \boldsymbol{v}(t,p) = g(p) \boldsymbol{u}(t),\quad g(p)=\e^{-p},p\ge 0,
\end{equation}
together with a symmetric extension of the initial data to region $p<0$,  system \eqref{equ:ODElinear} can be transformed to a system of linear convection equations \cite{jinQuantumSimulationPartial2023,jinQuantumSimulationPartial2024a}:
\begin{equation}\label{equ:u2v}
\begin{cases}
 \frac{\d}{\d t} \boldsymbol{v}(t,p) = A \boldsymbol{v}(t,p) = - H_1 \partial_p \boldsymbol{v} + \i H_2 \boldsymbol{v}, \\
 \boldsymbol{v}(0,p) = \e^{-|p|} \boldsymbol{u}_0.
 \end{cases}
\end{equation}

For numerical implementation, it is natural and convenient to introduce a function $\alpha = \alpha(p)$ into the initial data of \eqref{equ:u2v} for $p<0$, so that:
\begin{equation}\label{equ:u2valpha}
\begin{cases}
 \frac{\d }{\d  t} \bm{v}(t,p) = A \bm{v}(t,p) = - H_1 \partial_p \bm{v} + \i H_2 \bm{v}, \\
 \bm{v}(0,p) = \e^{-\alpha(p) |p|} \bm{u}_0.
 \end{cases}
\end{equation}
To recover the exact solution, it is necessary to impose $\alpha(p) = 1$ in the region $p> 0$. In the domain $p>0$, we truncate the domain at $p=R$, where $R$ is chosen sufficiently large so that $e^{-R} \approx 0$.  For $p<0$, we choose a large of $\alpha$ so that the solution is supported within a relatively small domain.  

Let $\tilde{\boldsymbol{v}} = Q^{-1} \boldsymbol{v}$.  Since each wave component $\tilde{\bm{v}}_j$ moves to the left, we introduce an artificial boundary at $p=L<0$, with $|L|$ sufficiently large so that $\tilde{\bm{v}}_j$, which is initially almost compact  supported in $[L_0, R]$, does not reach $p=L$ during the during the computational time interval $[0,T]$. This makes it possible to impose periodic boundary conditions in the $p$-direction for the spectral approximation. 

The above extension contributes to the first order accuracy with respect to $\Delta p$. Higher order extensions can be found in \cite{jinSchrodingerizationBasedQuantumAlgorithms2025,jinSchrodingerizationBasedComputationally2025} by using smooth initial data in $p$.
When $H_1$ has positive eigenvalues, we need to use the largest right and left moving speed
\begin{gather}
    \lambda_{\mathrm{max}}^{+}(H_{1})=\max\left\{\sup_{0<t<T}\{|\lambda|:\lambda\in\sigma(H_1(t)),\lambda>0\},0\right\},\\
    \lambda_{\mathrm{max}}^{-}(H_{1})=\max\left\{\sup_{0<t<T}\{|\lambda|:\lambda\in\sigma(H_1(t)),\lambda<0\},0\right\},
\end{gather}
as is defined in \cite{jinSchrodingerizationBasedQuantumAlgorithms2025}, to give the choice of $L$ and $R$. Specifically, $L$ and $R$ should satisfy 
\begin{equation}
    e^{-R+\lambda_{\max}^{+}(H_{1})T}\lesssim\epsilon\mathrm{~and~}e^{L+\lambda_{\max}^{-}(H_{1})T}\lesssim\epsilon
\end{equation}
where $\epsilon$ is the desired accuracy.

We now explain how to recover the solution $\bm{u}(t)$. According to the relation $\bm{u}(t)=\e^p \bm{v}(t,p)$ for all $p>0$, a straightforward approach is to select an arbitrary point $p^\diamond>0$ and define
\begin{equation}\label{equ:restore1point}
\bm{u}(t) = \e^{p^\diamond} \bm{v}(t,p^{\diamond}).
\end{equation}
When $H_1$ has positive eigenvalues, Theorem 3.1 in \cite{jinSchrodingerizationBasedQuantumAlgorithms2025} provides the condition for the choosing the recovery point, namely, $p^{\diamond}\geq\lambda_{\mathrm{max}}^{+}(H_{1})T$.

An alternative and more intuitive perspective is obtained by discretizing the $p$-domain and assembling the corresponding function values at different grid points. To this end, we introduce a uniform mesh with spacing $\Delta p=(R-L)/N_p$ for the auxiliary variable $p$, where $N_p=2N$ is taken to be even, and denote the grid points by $L = p_0<p_1<\cdots<p_{N_p} = R$.
To approximate $\bm{v}(t,p)$, we define the vector $\bm{w}$ by collecting the values of $\bm{v}$ at these grid points:
\[\bm{w} = [\bm{w}_1; \bm{w}_2; \cdots; \bm{w}_n], \]
with ``;'' denotes the concatenation of $\{\bm{w}_i\}_{i\ge 1}$ into a column vector. Equivalently, this representation can be viewed as a superposition state with respect to the basis $\ket{k}$:
\[
 \bm{w}_i = \sum_k \bm{v}_i (t,p_k) \ket{k},
\]
where $\bm{v}_i$ denotes the $i$th component of $\bm{v}$.

Applying the Fourier spectral method in the $p$-direction yields
\begin{equation}\label{equ:v2w}
\frac{\d}{\d t} \boldsymbol{w}(t) = -\i ( H_1 \otimes P_\mu ) \boldsymbol{w} + \i (H_2 \otimes I) \boldsymbol{w}.
\end{equation}
Here, $P_\mu$ denotes the matrix representation of the momentum operator $-\i\partial_p$, namely,
\[P_\mu = \Phi D_\mu \Phi^{-1},  \qquad D_\mu = \text{diag}(\mu_{-N}, \cdots, \mu_{N-1}), \]
where $\mu_l = 2\pi l/(R-L)$ are the Fourier modes. Moreover,
\[\Phi = (\phi_{jl})_{N_p\times N_p} = (\phi_l(x_j))_{N_p\times N_p}, \qquad \phi_l(x) = \e ^{\i \mu_l (x-L)} \]
is the matrix representation of the discrete Fourier transform.

At this point, the dynamics have been successfully recast as a Hamiltonian system.
Furthermore, by introducing the change of variables $\tilde{\bm{w}} = (I \otimes \Phi^{-1})\bm{w}$, one obtain
\begin{equation}\label{equ:generalSchr}
\frac{\d}{\d t} \tilde{\boldsymbol{w}}(t) = -\i ( H_1 \otimes D_\mu ) \tilde{\boldsymbol{w}} + \i (H_2 \otimes I) \tilde{\boldsymbol{w}}.
\end{equation}
Thanks to its unitary evolution, this is more amenable to an approximation by a quantum algorithm. In particular, if $H_1$ and $H_2$ are sparse, then \eqref{equ:generalSchr} is a unitary dynamics with the Hamiltonian $H = H_1 \otimes D_\mu - H_2 \otimes I$ that inherits the sparsity and has the following properties  
\[\begin{aligned}s(H)=\mathcal{O}(s(A)),\quad\|H\|_{\max}\leq\|H_1\|_{\max}/\Delta p+\|H_2\|_{\max},\end{aligned}\]
where $s(H)$ is the sparsity (maximum number of nonzero entries in each row and column) of the matrix $H$ and $\|H\|_{\max}$ is its max-norm (value of largest entry in absolute value). A direct Hamiltonian simulation can give the final quantum state solution $\tilde{\boldsymbol{w}}(T)$.

\subsection{Complexity analysis}

We use  Theorem 3 in \cite{lowOptimalHamiltonianSimulation2017} to compute the query complexity.
\begin{lemma}\label{lem:complexity}
    A $d$-sparse Hamiltonian $H$ on $n$ qubits with matrix elements specified to $m$ bits of precision can be simulated for time interval $[0,t]$, error $\epsilon$, and success probability at least $1 - 2\epsilon$  with $\mathcal{O}[td\|H\|_{\max}+\log{(1/\epsilon)}/\log{\log{(1/\epsilon)}}]$ queries to a unitary oracle $O$ that provides a description of $H$ and a factor $\mathcal{O}[n+m\mathrm{polylog}(m)]$ additional quantum gates.
\end{lemma}

\subsection{Example: the linear advection with interface}\label{sec:lin_advec}
In this subsection, we rewrite an example in \cite{jinQuantumSimulationPartial2024}. 
Although the result doesn't change, the refomulation are more general and can be applied in later sections.

We are concerned with a hyperbolic equation with discontinuous coefficients
\begin{equation}\label{equ:lin_advec_d-dim}
    \begin{cases}\partial_tu+\nabla\cdot(c(\boldsymbol{x})u)=0,\quad t>0,\boldsymbol{x}\in(-a,a)^d,\\u(0,\boldsymbol{x})=u_{in}(\boldsymbol{x}),\quad \boldsymbol{x}\in[-a,a]^d,\end{cases}
\end{equation}
where $a > 0$ is a constant and $c(\boldsymbol{x})=[c_1(\boldsymbol{x}),\cdotp\cdotp\cdotp,c_d(\boldsymbol{x})]^T$ is a vector for fixed $\boldsymbol{x}=(x_1,\cdots,x_d)$. We assume that $c (\boldsymbol{x})$ is piecewise constant in the $x_1$-direction:
\begin{equation*}
    c_i(\boldsymbol{x})=\begin{cases}c^-,&x_1<0\\c^+,&x_1>0\end{cases},\quad i=1,\cdots,d.
\end{equation*}
The above equation arises in modeling wave propagation through interfaces with jumps in $c (\boldsymbol{x})$ corresponding to interfaces between different media. To make the problem physically meaningful, $c^+$ and $c^-$ must have the same sign. For such a problem an interface condition is needed at $x_1 = 0$:
\begin{equation}\label{equ:lin_advec_cond}
    u(t,\boldsymbol{x}^+)=\rho u(t,\boldsymbol{x}^-),
\end{equation}
where $\boldsymbol{x}^{\pm}$ represents the right and left limits in the $x_1$-direction, $\rho = 1$ corresponds to the continuity of mass $u$ or $\rho = c^-/c^+$ for the continuity of flux $cu$.

In the following, we only consider the 1-D case:
\begin{equation}\label{equ:lin_advec_1-dim}
    \begin{cases}\partial_tu+c(x)\partial_xu=0,\quad t>0,x\in(-a,a),\\u(0,x)=u_{in}(x),\quad x\in[-a,a].\end{cases}
\end{equation}
The exact solution of \eqref{equ:lin_advec_1-dim} with the interface condition \eqref{equ:lin_advec_cond} can be constructed following characteristics, given by
\begin{equation*}
    u(t,x;u_0)=\begin{cases}u_0(x-c^-t),&\quad x<0,\\\rho u_0(\frac{c^-}{c^+}x-c^-t),&\quad0<x<c^+t,\\u_0(x-c^+t),&\quad x>c^+t,\end{cases}
\end{equation*}
when $c^+,c^->0$, or 
\begin{equation*}
    u(t,x;u_0)=\begin{cases}u_0(x-c^-t),&\quad x<c^-t,\\\frac{1}{\rho} u_0(\frac{c^+}{c^-}x-c^+t),&\quad c^-t<x<0,\\u_0(x-c^+t),&\quad x>0,\end{cases}
\end{equation*}
when $c^+,c^-<0$.
For different cases, we impose the boundary condition on the upwind endpoint.

When numerically solving \eqref{equ:lin_advec_1-dim}, the most natural approach is to build the interface condition \eqref{equ:lin_advec_cond} into the numerical flux. Let the uniform spatial mesh be $x_{j}=j\Delta x$, where $j=-N_x,\cdotp\cdotp\cdotp,-1,0,1,\cdotp\cdotp\cdotp,N_x+1$ and $\Delta x=a/N_x$ is the mesh size. Here we add the point $x_{N_x+1}$ to have $2N_x$ variables after applying the boundary conditions. We can apply the upwind scheme for $x\le-\Delta x$ and $x>\Delta x$:
\begin{equation*}
    \begin{cases}
    \dfrac{\mathrm{d}u_j}{\mathrm{d}t}=-\dfrac{-\frac{|c^-|+c^-}{2}u_{j-1}(t)+|c^-|u_j(t)-\frac{|c^-|-c^-}{2}u_{j+1}(t)}{\Delta x},&\quad j=-(N_x-1),\cdots,-1,\\
    \dfrac{\mathrm{d}u_j}{\mathrm{d}t}=-\dfrac{-\frac{|c^+|+c^+}{2}u_{j-1}(t)+|c^+|u_j(t)-\frac{|c^+|-c^+}{2}u_{j+1}(t)}{\Delta x},&\quad j=2,\cdots,N_x.
    \end{cases}
\end{equation*}
Note that in the above scheme $u_0$ is the left limit of $u$ at the interface. For $j = 0$, the continuity of $cu$ gives
\begin{equation*}
    \frac{\mathrm{d}u_0}{\mathrm{d}t}=-\frac{-\frac{|c^-|+c^-}{2}u_{-1}(t)+|c^-|u_0(t)-\frac{|c^+|-c^+}{2}u_{1}(t)}{\Delta x}.
\end{equation*}
For $j = 1$, the continuity of $cu$ gives
\begin{equation*}
    \frac{\mathrm{d}u_1}{\mathrm{d}t}=-\frac{-\frac{|c^-|+c^-}{2}u_{0}(t)+|c^+|u_1(t)-\frac{|c^+|-c^+}{2}u_{2}(t)}{\Delta x}.
\end{equation*}
Let $\boldsymbol{u} (t)=[u_{-(N_x-1)}(t),\cdots, u_{N_x}(t)]^T$. One can collect the above equations as the system \eqref{equ:ODElinear} with
\vspace{0.7em}
\begin{equation}
    A=\frac{1}{\Delta x}  
\left[\begin{array}{ccccc:ccccc} 
\tikzmark{upper1}-|c^-|&\frac{|c^-|-c^-}{2}&&&\tikzmark{upper2}&\tikzmark{upper3}&&&&\tikzmark{upper4}\\
\frac{|c^-|+c^-}{2}&-|c^-|&\frac{|c^-|-c^-}{2}&&&&&&&\\
&\ddots&\ddots&\ddots&&&&&& \\
&  & \frac{|c^-|+c^-}{2} & -|c^-|&\frac{|c^-|-c^-}{2}&&&&&\\
&&&\frac{|c^-|+c^-}{2} & -|c^-|&\frac{|c^+|-c^+}{2}&&&&\\
\hdashline
&&&&\frac{|c^-|+c^-}{2}&-|c^+|&\frac{|c^+|-c^+}{2}&&&\\
&&&&&\frac{|c^+|+c^+}{2}&-|c^+|&\frac{|c^+|-c^+}{2}&&\\
&&&&&&\ddots&\ddots&\ddots&\\
&&&&&&&\frac{|c^+|+c^+}{2}&-|c^+|&\frac{|c^+|-c^+}{2}\\
&&&&&&&&\frac{|c^+|+c^+}{2}&-|c^+|
    \end{array}\right]
\begin{tikzpicture}[overlay, remember picture,decoration={brace,amplitude=5pt}]
\draw[decorate,thick] (upper1.north) -- (upper2.north)
      node [midway,above=0.5pt] {$N_x$ columns};
\draw[decorate,thick] (upper3.north) -- (upper4.north)
      node [midway,above=0.5pt] {$N_x$ columns};
\end{tikzpicture}
\end{equation}
and $\boldsymbol{b}(t)=[\frac{|c^-|+c^-}{2}u(t,-a)/\Delta x,0,\cdotp\cdotp\cdotp,0,\frac{|c^+|-c^+}{2}u(t,a)/\Delta x]^T.$

This simple example illustrates that the absolute value function can be used to automatically derive the matrix form of the upwind scheme in both cases $c(x)>0$ and $c(x)<0$. Building on this observation, we introduce more sophisticated nonlinear functions to obtain the matrix representation associated with the Hamiltonian-preserving scheme.

A numerical test based on Schr\"odingerization is conducted in \cite{jinQuantumSimulationPartial2024}. We will not repeat it here.

\section{Quantum simulation of Liouville equation with interface condition \eqref{equ:ddimLiouvilleOpticsinterfacecondition} in one space dimension}\label{sec:interface1}

We consider the numerical solution of the Liouville equation in one physical space dimension
\begin{equation}\label{equ:1dimLiouvilleOptics}
    f_t+c(x)\operatorname{sign}(\xi)f_x-c_x|\xi|f_\xi=0,
\end{equation}
with a discontinuous wave speed $c(x)$ and constant Hamiltonian condition at the interface:
\begin{equation}\label{equ:1dimLiouvilleOpticsPreserve}
    c^+\xi^+=c^-\xi^-.
\end{equation}
Then the interface condition \eqref{equ:ddimLiouvilleOpticsinterfaceconditionPartial} becomes 
\begin{equation}\label{equ:1dimLiouvilleOpticsinterfaceconditionPartial}
    f(t,x^+,\xi^+)=\alpha^Tf(t,x^-,\xi^-)+\alpha^Rf(t,x,-\xi^+).
\end{equation}

\subsection{The original algorithm}
We employ a uniform mesh with grid points at $x_{i+\frac{1}{2}}$, $i=0,\ldots,N_x$, in the $x$-direction and $\xi_{j+\frac{1}{2}}$, $j=0,\ldots,N_\xi$ in the $\xi$-direction. The cells are centered at $(x_i,\xi_j)$, $i=1,\ldots,N_x$, $j=1,\ldots,N_\xi$ with $x_i=\frac{1}{2}(x_{i+\frac{1}{2}}+x_{i-\frac{1}{2}})$ and $\xi_j=\frac{1}{2}(\xi_{j+\frac{1}{2}}+\xi_{j-\frac{1}{2}})$. The mesh size is denoted by $\Delta x= x_{i+\frac{1}{2}}-x_{i-\frac{1}{2}}$, $\Delta \xi=\xi_{j+\frac{1}{2}}-\xi_{j-\frac{1}{2}}$. We define the cell average of $f$ as 
\begin{equation*}
    f_{ij}=\frac{1}{\Delta x\Delta \xi}\int_{x_{i-\frac{1}{2}}}^{x_{i+\frac{1}{2}}}\int_{\xi_{i-\frac{1}{2}}}^{\xi_{i+\frac{1}{2}}}f(x,\xi,t)d\xi dx.
\end{equation*}

Assume that the discontinuous points of wave speed $c(x)$ are located at the grid points. Let the left and right limits of $c(x)$ at point $x_{i+1/2}$ be $c^{-}_{i+1/2}$ and $c^{+}_{i+1/2}$, respectively. Note that if $c$ is continuous at $x_{i+1/2}$, then $c^{+}_{i+1/2}=c^{-}_{i+1/2}$. We approximate $c$ by a piecewise linear function
\begin{equation*}c(x)\approx c_{i-1/2}^++\frac{c_{i+1/2}^--c_{i-1/2}^+}{\Delta x}(x-x_{i-1/2}).\end{equation*}
We also define the averaged wave speed as $c_i = \frac{c_{i-\frac{1}{2}}^+ + c_{i+\frac{1}{2}}^-}{2}$. The semi-discrete scheme (with continuous time) reads
\begin{equation}\label{equ:1dimLiouvilleOpticsSmeidis}
    (f_{ij})_{t}+\frac{c_{i}\operatorname{sign}(\xi_{j})}{\Delta x}\left(f_{i+\frac{1}{2},j}^{-}-f_{i-\frac{1}{2},j}^{+}\right)-\frac{c_{i+\frac{1}{2}}^{-}-c_{i-\frac{1}{2}}^{+}}{\Delta x\Delta\xi}|\xi_{j}|\left(f_{i,j+\frac{1}{2}}-f_{i,j-\frac{1}{2}}\right)=0,
\end{equation}
where the numerical fluxes $f_{i,j+\frac{1}{2}}$ is defined using the upwind discretization.  Since the characteristics of the Liouville equation may be different on the two sides of the interface, the corresponding numerical fluxes should also be different. Then the essential part of Hamiltonian-preserving schemes is to use \eqref{equ:1dimLiouvilleOpticsinterfaceconditionPartial} to define the split numerical fluxes $f^-_{i+\frac{1}{2},j},f^+_{i-\frac{1}{2},j}$ at each interface.

Assume $c$ is discontinuous at $x_{i+\frac{1}{2}}$. Consider the case $\xi_j > 0$. Using upwind scheme, $f_{i+\frac12,j}^- = f_{ij}$. However, $f_{i+\frac{1}{2},j}^+=\alpha^Tf(t,x_{i+\frac{1}{2}}^-,\xi^-)+\alpha^Rf(t,x_{i+\frac{1}{2}}^+,-\xi^+)$ while $\xi^-$ is obtained from $\xi_j^+ = \xi_j$ from \eqref{equ:1dimLiouvilleOpticsPreserve}. Since $\xi^-$ may not be a grid point, one has to define it approximately. One can first locate the two cell centers that bound $\xi^-$, then use a linear interpolation to evaluate the needed numerical flux. The case $\xi_j < 0$ is treated similarly. The detailed algorithm to generate the numerical flux is given in Algorithm \ref{alg:fluxOptics1tranPartial}.

\begin{algorithm}[!htb]
  \label{alg:fluxOptics1tranPartial}
  \caption{Computation of the numerical flux in \eqref{equ:1dimLiouvilleOpticsSmeidis}}
  \SetAlgoLined
  \KwIn{$\xi_j$, $c^-_{i+\frac{1}{2}}$, $c^+_{i+\frac{1}{2}}$, $\{f_{ij},j=1,\ldots,N_\xi\}$ and $\{f_{i+1,j},j=1,\ldots,N_\xi\}$  }
  \KwOut{$f^-_{i+\frac{1}{2},j}$ and $f^+_{i+\frac{1}{2},j}$}
  \If{$\xi_j>0$}{
      $f^-_{i+\frac{1}{2},j}=f_{ij}$\;
      $\xi^-=\frac{c^+_{i+\frac{1}{2}}}{c^-_{i+\frac{1}{2}}}\xi_j$\;
      \If{$\xi_k\le\xi^-<\xi_{k+1}$ for some $k$}
      {$\alpha^{R}_{i+\frac{1}{2}}=\left(\frac{c_{i+\frac{1}{2}}^{+}-c_{i+\frac{1}{2}}^{-}}{c_{i+\frac{1}{2}}^{+}+c_{i+\frac{1}{2}}^{-}}\right)^{2},\quad\alpha^{T}_{i+\frac{1}{2}}=1-\alpha^{R}_{i+\frac{1}{2}}$\;
      $f^+_{i+\frac12,j} = \alpha^T_{i+\frac{1}{2}}\left(\frac{\xi_{k+1}-\xi^-}{\Delta\xi}f_{i,k} + \frac{\xi^--\xi_k}{\Delta\xi}f_{i,k+1}\right)+\alpha^R_{i+\frac{1}{2}}f_{i+1,j'}$ where $\xi_{j'}=-\xi_j$\;}
      }
  \If{$\xi_j<0$}{
      $f^+_{i+\frac{1}{2},j}=f_{i+1,j}$\;
      $\xi^+=\frac{c^-_{i+\frac{1}{2}}}{c^+_{i+\frac{1}{2}}}\xi_j$\;
      \If{$\xi_k\le\xi^+<\xi_{k+1}$ for some $k$}
      {$\alpha^{R}_{i+\frac{1}{2}}=\left(\frac{c_{i+\frac{1}{2}}^{+}-c_{i+\frac{1}{2}}^{-}}{c_{i+\frac{1}{2}}^{+}+c_{i+\frac{1}{2}}^{-}}\right)^{2},\quad\alpha^{T}_{i+\frac{1}{2}}=1-\alpha^{R}_{i+\frac{1}{2}}$\;
      $f^-_{i+\frac12,j} =  \alpha^T_{i+\frac{1}{2}}\left(\frac{\xi_{k+1}-\xi^+}{\Delta\xi}f_{i+1,k} + \frac{\xi^+-\xi_k}{\Delta\xi}f_{i+1,k+1}\right)+\alpha^R_{i+\frac{1}{2}}f_{i,j'}$ where $\xi_{j'}=-\xi_j$\;}
      }
\end{algorithm}


\begin{remark}\label{rmk:convergencerate}
    Although the Algorithm \ref{alg:fluxOptics1tranPartial} for evaluating numerical flux is of first order, the convergence rate of the numerical solution in $l^1$ norm may be less than first order. As stated by the well-established theory \cite{kuznetsov1977stable,tangSharpnessKuznetsovsOsqrtDelta1995}, when using a usual finite difference method for solving the discontinuous solution of a linear hyperbolic equation, the convergence rate is at most $1/2$ order. However, when the only discontinuity in the solutions is at the interface, the convergence rate of Hamiltonian-preserving scheme can still be first order \cite{jinHamiltonianPreservingSchemeLiouville2006}.
\end{remark}

\subsection{The ODE form}\label{sec:interface1ODEform}
In the following, we convert the semi-discrete scheme \eqref{equ:1dimLiouvilleOpticsSmeidis} and Algorithm \ref{alg:fluxOptics1tranPartial} into the ODE form \eqref{equ:ODElinear}. Assume that $N_x$ and $N_\xi$ are even numbers and the truncated interval in $x$ and $\xi$-direction are symmetric about $0$. (For asymmetrical case, the matrix $A$ in ODE \eqref{equ:ODElinear} will be slightly modified.) Then $\xi_j<0$ for $j=1,\ldots,\frac{N_\xi}{2}$ and $\xi_j>0$ for $j=\frac{N_\xi}{2}+1,\ldots,N_\xi$. We also assume that interfaces don't appear at $x_\frac{1}{2}$ and $x_{N_x+\frac12}$. Denote $\boldsymbol{f}(t)=[f_{11},\ldots,f_{1,N_\xi},f_{21},\ldots,f_{2,N_\xi},\ldots,f_{N_x,1}, \ldots,f_{N_x,N_\xi}]^T=\sum_{i=1}^{N_x}\sum_{j=1}^{N_\xi}f_{ij}(t)\ket{i-1}\ket{j-1}$. This definition specifies not only the order of the equations but also the order of the variables. We can also arrange $\boldsymbol{f}$ in different orders, then we only need to change the corresponding columns and rows of $A$ to get the new matrix.

If $\xi_j>0$, for the term $\frac{c_{i}\operatorname{sign}(\xi_{j})}{\Delta x}\left(f_{i+\frac{1}{2},j}^{-}-f_{i-\frac{1}{2},j}^{+}\right)$ in \eqref{equ:1dimLiouvilleOpticsSmeidis}, $\operatorname{sign}(\xi_{j})=1$, $f_{i+\frac{1}{2},j}^{-}=f_{ij}$ and $f^+_{i-\frac12,j} = \alpha^T_{i-\frac{1}{2}}\left(\frac{\xi_{k+1}-\xi^-}{\Delta\xi}f_{i-1,k} + \frac{\xi^--\xi_k}{\Delta\xi}f_{i-1,k+1}\right)+\alpha^R_{i-\frac{1}{2}}f_{i,j'}$ where $\xi_k\le\xi^-=\frac{c^+_{i-\frac{1}{2}}}{c^-_{i-\frac{1}{2}}}\xi_j<\xi_{k+1}$ for some $k$ and $\xi_{j'}=-\xi_j$. The search of which interval $\xi^-$ lies in involves a nonlinear operation. We need to use nonlinear functions to encode this nonlinearity into the matrix $A$. Define the hat function 
\begin{equation}\label{equ:hat}
    h(z)=\max\left(1-\frac{|z|}{\Delta\xi},0\right).
\end{equation}
The hat function can be regarded as a ``distance" function because it gives a non-negative value according to the distance from zero. It has a compact support of size $2\Delta\xi$ on $\mathbb{R}$.
Then 
\[
f^+_{i-\frac12,j} =\sum_{k=1}^{N_\xi}\alpha^T_{i-\frac{1}{2}}h(\xi^--\xi_k)f_{i-1,k}+\alpha^R_{i-\frac{1}{2}}f_{i,j'}=\sum_{k=1}^{N_\xi}\alpha^T_{i-\frac{1}{2}}h\left(\frac{c^+_{i-\frac{1}{2}}}{c^-_{i-\frac{1}{2}}}\xi_j-\xi_k\right)f_{i-1,k}+\alpha^R_{i-\frac{1}{2}}f_{i,j'}
\]
and 
\begin{equation*}
    \frac{c_{i}\operatorname{sign}(\xi_{j})}{\Delta x}\left(f_{i+\frac{1}{2},j}^{-}-f_{i-\frac{1}{2},j}^{+}\right)=\frac{c_{i}}{\Delta x}f_{ij}-\frac{c_{i}}{\Delta x}\sum_{k=1}^{N_\xi}\alpha^T_{i-\frac{1}{2}}h\left(\frac{c^+_{i-\frac{1}{2}}}{c^-_{i-\frac{1}{2}}}\xi_j-\xi_k\right)f_{i-1,k}-\frac{c_{i}}{\Delta x}\alpha^R_{i-\frac{1}{2}}f_{i,j'}.
\end{equation*}

If $\xi_j<0$, for the term $\frac{c_{i}\operatorname{sign}(\xi_{j})}{\Delta x}\left(f_{i+\frac{1}{2},j}^{-}-f_{i-\frac{1}{2},j}^{+}\right)$ in \eqref{equ:1dimLiouvilleOpticsSmeidis}, similarly, we can get 
\begin{equation*}
    \frac{c_{i}\operatorname{sign}(\xi_{j})}{\Delta x}\left(f_{i+\frac{1}{2},j}^{-}-f_{i-\frac{1}{2},j}^{+}\right)=-\frac{c_{i}}{\Delta x}\sum_{k=1}^{N_\xi}\alpha^T_{i+\frac{1}{2}}h\left(\frac{c^-_{i+\frac{1}{2}}}{c^+_{i+\frac{1}{2}}}\xi_j-\xi_k\right)f_{i+1,k}-\frac{c_{i}}{\Delta x}\alpha^R_{i+\frac{1}{2}}f_{i,j'}+\frac{c_{i}}{\Delta x}f_{ij}.
\end{equation*}

Define 
\begin{equation}\label{equ:beta}
    \beta_{i+\frac12,j,k}=\left\{\begin{array}{ll}
    h\left(\frac{c^+_{i+\frac{1}{2}}}{c^-_{i+\frac{1}{2}}}\xi_j-\xi_k\right),     & i=0,\ldots,N_x,\quad j=\frac{N_\xi}{2}+1,\ldots,N_\xi,\quad k=1,\ldots N_\xi, \\
    h\left(\frac{c^-_{i+\frac{1}{2}}}{c^+_{i+\frac{1}{2}}}\xi_j-\xi_k\right),     & i=0,\ldots,N_x,\quad j=1,\ldots,\frac{N_\xi}{2},\quad k=1,\ldots N_\xi. 
    \end{array}\right.
\end{equation}
Here we use $i+\frac12$ to show that $\beta_{i+\frac12,j,k}$ is related to the coefficients $c^-_{i+\frac{1}{2}}$ and $c^+_{i+\frac{1}{2}}$ of the interface at $x_{i+\frac{1}{2}}$.
Then
\begin{equation*}
    \frac{c_{i}\operatorname{sign}(\xi_{j})}{\Delta x}\left(f_{i+\frac{1}{2},j}^{-}-f_{i-\frac{1}{2},j}^{+}\right)=\left\{\begin{array}{ll}
    \frac{c_{i}}{\Delta x}f_{ij}-\frac{c_{i}}{\Delta x}\sum_{k=1}^{N_\xi}\alpha^T_{i-\frac{1}{2}}\beta_{i-\frac12,j,k}f_{i-1,k}-\frac{c_{i}}{\Delta x}\alpha^R_{i-\frac{1}{2}}f_{i,j'},     \\ \qquad\qquad i=1,\ldots,N_x,\quad j=\frac{N_\xi}{2}+1,\ldots,N_\xi, \\
    \frac{c_{i}}{\Delta x}f_{ij}-\frac{c_{i}}{\Delta x}\sum_{k=1}^{N_\xi}\alpha^T_{i+\frac{1}{2}}\beta_{i+\frac12,j,k}f_{i+1,k}-\frac{c_{i}}{\Delta x}\alpha^R_{i+\frac{1}{2}}f_{i,j'},     \\ \qquad\qquad i=1,\ldots,N_x,\quad j=1,\ldots,\frac{N_\xi}{2}.
    \end{array}
    \right.
\end{equation*}
When $c^-_{i+\frac{1}{2}}=c^+_{i+\frac{1}{2}}$ the above scheme degenerates to the upwind scheme. 
Since we have assumed that interfaces don't appear at $x_\frac{1}{2}$ and $x_{N_x+\frac12}$, we can use the boundary condition to give the value of $f_{0j}(t)$ and $f_{N_x+1,j}(t)$.
When $\xi_j>0$ (resp. $\xi_j<0$), we impose inflow boundary condition at $(x_0,\xi_j)$ (resp. $(x_{N_x+1},\xi_j)$) and outflow boundary condition at $(x_{N_x+1},\xi_j)$ (resp. $(x_0,\xi_j)$).
Write the above equations in the form $A_1\boldsymbol{f}+b_1$ of a product of a matrix and a vector with an inhomogeneous term representing boundary condition, where the equations and variables are in order of $\boldsymbol{f}$, we can get a $N_xN_\xi\times N_xN_\xi$ block tridiagonal matrix
\begin{equation}\label{equ:blocktridiagmtx}
    A_1=\frac{1}{\Delta x}
    \begin{pmatrix}
        C_1 & -c_1B^u_1 &&&&&\\
        -c_2B^l_1 & C_2 & -c_2B^u_2 &&&&\\
        & -c_3B^l_2 & C_3 & -c_3B^u_3 &&&\\
        && \ddots & \ddots & \ddots &&\\
        &&& -c_{N_x-1}B^l_{N_x-2} &C_{N_x-1}& -c_{N_x-1}B^u_{N_x-1}\\
        &&&& -c_{N_x}B^l_{N_x-1} &C_{N_x}\\
    \end{pmatrix}_{N_xN_\xi\times N_xN_\xi}
\end{equation}
and a $N_x N_\xi$-dimensional vector 
\begin{equation}\label{equ:inhomogeneoustermx}
    b_1=-\frac{c_1}{\Delta x}\sum_{j=\frac{N_\xi}{2}+1}^{N_\xi}f_{0j}(t)\ket{0}\ket{j-1}-\frac{c_{N_x}}{\Delta x}\sum_{j=1}^{\frac{N_\xi}{2}}f_{N_x+1, j}(t)\ket{N_x-1}\ket{j-1},
\end{equation}
where
\begin{equation}
    C_i=c_iI_{N_\xi\times N_\xi}+c_i\begin{pmatrix}
        & & & &  & -\alpha^R_{i+\frac12}\\
         & & &  & \iddots &\\
         & &  & -\alpha^R_{i+\frac12} & &\\
         &  & -\alpha^R_{i-\frac12} & & &\\
        & \iddots & & & &  \\
        -\alpha^R_{i-\frac12} &  & & & & 
    \end{pmatrix}_{N_\xi\times N_\xi},
\end{equation}
$B_i$'s $(j,k)$ element is $\alpha^T_{i+\frac12}\beta_{i+\frac12,j,k}$ and $B^u_i$ (resp. $B^l_i$) refers to taking the upper (resp. lower) half of the matrix $B_i$ and making the rest $0$, i.e., 
\begin{equation*}
    B_i=\alpha^T_{i+\frac12}\begin{pmatrix}
        \beta_{i+\frac12,1,1}&\beta_{i+\frac12,1,2}&\cdots&\beta_{i+\frac12,1,N_\xi}\\
        \beta_{i+\frac12,2,1}&\beta_{i+\frac12,2,2}&\cdots&\beta_{i+\frac12,2,N_\xi}\\
        \vdots&\vdots&\ddots&\vdots\\
        \beta_{i+\frac12,N_\xi,1}&\beta_{i+\frac12,N_\xi,2}&\cdots&\beta_{i+\frac12,N_\xi,N_\xi}\\
    \end{pmatrix}_{N_\xi\times N_\xi},
\end{equation*}
and
\begin{equation*}
    B_i^u=\alpha^T_{i+\frac12}\begin{pmatrix}
        \beta_{i+\frac12,1,1}&\beta_{i+\frac12,1,2}&\cdots&\beta_{i+\frac12,1,N_\xi}\\
        \vdots&\vdots&\ddots&\vdots\\
        \beta_{i+\frac12,\frac{N_\xi}{2},1}&\beta_{i+\frac12,\frac{N_\xi}{2},2}&\cdots&\beta_{i+\frac12,\frac{N_\xi}{2},N_\xi}\\
        0&0&\cdots&0\\
        \vdots&\vdots&\ddots&\vdots\\
        0&0&\cdots&0\\
    \end{pmatrix}_{N_\xi\times N_\xi},
\end{equation*}
\begin{equation*}
    B_i^l=\alpha^T_{i+\frac12}\begin{pmatrix}
        0&0&\cdots&0\\
        \vdots&\vdots&\ddots&\vdots\\
        0&0&\cdots&0\\
        \beta_{i+\frac12,\frac{N_\xi}{2}+1,1}&\beta_{i+\frac12,\frac{N_\xi}{2}+1,2}&\cdots&\beta_{i+\frac12,\frac{N_\xi}{2}+1,N_\xi}\\
        \vdots&\vdots&\ddots&\vdots\\
        \beta_{i+\frac12,N_\xi,1}&\beta_{i+\frac12,N_\xi,2}&\cdots&\beta_{i+\frac12,N_\xi,N_\xi}\\
    \end{pmatrix}_{N_\xi\times N_\xi}.
\end{equation*}
To facilitate subsequent complexity analysis, we will now discuss the sparsity of $B_i^u$ and $B_i^l$. 
Hermitian matrix have the same sparsity in row and column. For general matrix, we give the following definition.
\begin{definition}(Sparsity in row or column)
    For a general matrix $A\in \mathbb{C}^{m\times n}$, the sparsity in row (i.e., the number of nonzero elements in row) is denoted as $s_r(A)$ and the sparsity in column (i.e., the number of nonzero elements in column) is denoted as $s_c(A)$. For a Hermitian matrix $A\in \mathbb{C}^{m\times m}$, $s(A):=s_r(A)=s_c(A)$.
\end{definition}
Note that since the hat function \eqref{equ:hat} has a compact support of size $2\Delta\xi$, $s_r(B_i)=s_r(B_i^u)=s_r(B_i^l)\le2$. 
$s_c(B_i^u)$ and $s_c(B_i^l)$ is dependent on the local wave speed along both sides of the interface.
\begin{itemize}
    \item If $\frac{c^+_{i+\frac{1}{2}}}{c^-_{i+\frac{1}{2}}}>1$ for some $i$, then $s_c(B_i^l)\le2$ and $\beta_{i+\frac12,j,k}=0$ for $j=\frac{N_\xi}{2}+1,\ldots,N_\xi,k=1,\ldots,\frac{N_\xi}{2}$, because if $\frac{c^+_{i+\frac{1}{2}}}{c^-_{i+\frac{1}{2}}}\xi_{j}$ lies in $[\xi_{k},\xi_{k+1})$, then $\frac{c^+_{i+\frac{1}{2}}}{c^-_{i+\frac{1}{2}}}\xi_{j+1}$ must lie in $[\xi_{k+l},\xi_{k+1+l})$ for $l\ge1$. 
    $s_c(B_i^u)\le2\left\lceil\frac{c^+_{i+\frac{1}{2}}}{c^-_{i+\frac{1}{2}}}\right\rceil$ and $\beta_{i+\frac12,j,k}=0$ for $j=1,\ldots,\frac{N_\xi}{2},k=\frac{N_\xi}{2}+2,\ldots,N_\xi$, since $[\xi_{k},\xi_{k+1})$ may contains at most $\left\lfloor\frac{c^+_{i+\frac{1}{2}}}{c^-_{i+\frac{1}{2}}}\right\rfloor$ rescaled intervals like $\left[\frac{c^-_{i+\frac{1}{2}}}{c^+_{i+\frac{1}{2}}}\xi_{j},\frac{c^-_{i+\frac{1}{2}}}{c^+_{i+\frac{1}{2}}}\xi_{j+1}\right)$ .
    \item If $\frac{c^-_{i+\frac{1}{2}}}{c^+_{i+\frac{1}{2}}}>1$ for some $i$, similarly, then $s_c(B_i^u)\le2$ and $\beta_{i+\frac12,j,k}=0$ for $j=1,\ldots,\frac{N_\xi}{2},k=\frac{N_\xi}{2}+1,\ldots,N_\xi$. 
    $s_c(B_i^l)\le 2\left\lceil\frac{c^-_{i+\frac{1}{2}}}{c^+_{i+\frac{1}{2}}}\right\rceil$ and $\beta_{i+\frac12,j,k}=0$ for $j=\frac{N_\xi}{2}+1,\ldots,N_\xi,k=1,\ldots,\frac{N_\xi}{2}-1$.
    \item If $c^-_{i+\frac{1}{2}}=c^+_{i+\frac{1}{2}}$ for some $i$, 
    \begin{equation*}
        B_i^u=\begin{pmatrix}
            I_{\frac{N_\xi}{2}\times\frac{N_\xi}{2}}&O_{\frac{N_\xi}{2}\times\frac{N_\xi}{2}}\\
            O_{\frac{N_\xi}{2}\times\frac{N_\xi}{2}}&O_{\frac{N_\xi}{2}\times\frac{N_\xi}{2}}
        \end{pmatrix},
        B_i^l=\begin{pmatrix}
            O_{\frac{N_\xi}{2}\times\frac{N_\xi}{2}}&O_{\frac{N_\xi}{2}\times\frac{N_\xi}{2}}\\
            O_{\frac{N_\xi}{2}\times\frac{N_\xi}{2}}&I_{\frac{N_\xi}{2}\times\frac{N_\xi}{2}}
        \end{pmatrix},
    \end{equation*}
    which degenerates to the ordinary upwind scheme and $s_c(B_i^u)=s_c(B_i^l)=1$.
\end{itemize}
Denote the set of interfaces as $\mathcal{I}=\{x_{i+\frac{1}{2}}:c(x)\text{ is discontinuous on }x_{i+\frac{1}{2}}\}$. Then $s_r(A_1)\le4$ and $s_c(A_1)\le2\max_{x\in\mathcal{I}}\left\{\left\lceil\frac{c^+(x)}{c^-(x)}\right\rceil,\left\lceil\frac{c^-(x)}{c^+(x)}\right\rceil\right\}+2:=Q$.

For the term $-\frac{c_{i+\frac{1}{2}}^{-}-c_{i-\frac{1}{2}}^{+}}{\Delta x\Delta\xi}|\xi_{j}|\left(f_{i,j+\frac{1}{2}}-f_{i,j-\frac{1}{2}}\right)$ in \eqref{equ:1dimLiouvilleOpticsSmeidis}, define $d_{ij}=-\frac{c_{i+\frac{1}{2}}^{-}-c_{i-\frac{1}{2}}^{+}}{\Delta x\Delta\xi}|\xi_{j}|$. 
Since the numerical fluxes $f_{i,j+\frac{1}{2}},f_{i,j-\frac{1}{2}}$ is defined using the upwind discretization, like the discussion in section \ref{sec:lin_advec},
\begin{equation*}
    d_{ij}\left(f_{i,j+\frac{1}{2}}-f_{i,j-\frac{1}{2}}\right)=-\frac{|d_{ij}|+d_{ij}}{2}f_{i,j-1}(t)+|d_{ij}|f_{ij}(t)-\frac{|d_{ij}|-d_{ij}}{2}f_{i,j+1}(t).
\end{equation*}
When $d_{ij}>0$ (resp. $d_{ij}<0$) for any $j$ we impose inflow boundary condition at $(x_i,\xi_0)$ (resp. $(x_i,\xi_{N_\xi+1})$) and outflow boundary condition at $(x_i,\xi_{N_\xi+1})$ (resp. $(x_i,\xi_0)$). 
Write the above equations in the form $A_2\boldsymbol{f}+b_2$ of a product of a matrix and a vector with an inhomogeneous term representing boundary condition, where the equations and variables are in order of $\boldsymbol{f}$, we can get a $N_xN_\xi\times N_xN_\xi$ block diagonal matrix 
\begin{equation}\label{equ:blockdiagmtx}
    A_2=\begin{pmatrix}
        D_1 &&&\\
        & D_2&&\\
        &&\ddots&\\
        &&&D_{N_x}
    \end{pmatrix}_{N_xN_\xi\times N_xN_\xi}
\end{equation}
and a $N_xN_\xi$-dimensional vector
\begin{equation}\label{equ:inhomogeneoustermxi}
    b_2=-\sum_{i=1}^{N_x}\frac{|d_{i1}|+d_{i1}}{2}f_{i0}(t)\ket{i-1}\ket{0}-\sum_{i=1}^{N_x}\frac{|d_{i,N_\xi}|-d_{i,N_\xi}}{2}f_{i,N_\xi+1}(t)\ket{i-1}\ket{N_\xi-1},
\end{equation}
where 
\begin{align*}
    D_i=&\begin{pmatrix}
        |d_{i1}|&-\frac{|d_{i1}|-d_{i1}}{2}&&&\\
        -\frac{|d_{i2}|+d_{i2}}{2}&|d_{i2}|&-\frac{|d_{i2}|-d_{i2}}{2}&&\\
        &\ddots&\ddots&\ddots&\\
        && -\frac{|d_{i,N_\xi-1}|+d_{i,N_\xi-1}}{2}&|d_{i,N_\xi-1}|&-\frac{|d_{i,N_\xi-1}|-d_{i,N_\xi-1}}{2}\\
        &&& -\frac{|d_{i,N_\xi}|+d_{i,N_\xi}}{2}&|d_{i,N_\xi}|
    \end{pmatrix}_{N_\xi\times N_\xi}\\
    =&\frac12\begin{pmatrix}
        d_{i1} &&&&\\
        & d_{i2}&&&\\
        && \ddots &&\\
        &&& d_{i,N_{\xi}-1}&\\
        &&&& d_{i,N_{\xi}}
    \end{pmatrix}
    \begin{pmatrix}
        0 & 1 &&&\\
        -1 &0&1&&\\
        & \ddots&\ddots&\ddots&\\
        &&-1&0&1\\
        &&&-1&0
    \end{pmatrix}\\
    &+\frac12\begin{pmatrix}
        |d_{i1}| &&&&\\
        & |d_{i2}|&&&\\
        && \ddots &&\\
        &&& |d_{i,N_{\xi}-1}|&\\
        &&&& |d_{i,N_{\xi}}|
    \end{pmatrix}
    \begin{pmatrix}
        2 & -1 &&&\\
        -1 &2&-1&&\\
        & \ddots&\ddots&\ddots&\\
        &&-1&2&-1\\
        &&&-1&2
    \end{pmatrix}.
\end{align*}

Combining \eqref{equ:1dimLiouvilleOpticsSmeidis}, \eqref{equ:blocktridiagmtx}, \eqref{equ:blockdiagmtx}, \eqref{equ:inhomogeneoustermx} and \eqref{equ:inhomogeneoustermxi}, we can get the ODE $\boldsymbol{f}'(t)=A\boldsymbol{f}(t)+\boldsymbol{b}(t)$ where $A=-A_1-A_2$ and $\boldsymbol{b}(t)=-b_1-b_2$.

\begin{remark}
    The algorithm \ref{alg:fluxOptics1tranPartial} for evaluating numerical fluxes is of first order. One can obtain a second order flux by incorporating the slope limiter, such as van Leer or minmod slope limiter, into the above algorithm. In a classical computer, this can be achieved by replacing $f_{ik}$ with $f_{ik}+\frac{\Delta x}2s_{ik}$ and replacing $f_{i+1,k}$ with $f_{i+1,k}-\frac{\Delta x}{2}s_{i+1,k}$, where $s_{ik}=\frac{1}{\Delta x}(f_{i+1,k}-f_{ik})\psi(\theta_i)$ is the slope limiter in the $x$-direction, $\psi(\theta)$ is a nonlinear function (for van Leer slope limiter, $\psi(\theta)=\frac{|\theta|+\theta}{1+|\theta|}$) and $\theta_i=\frac{f_{i,k}-f_{i-1,k}}{f_{i+1,k}-f_{i,k}}$.
    To compute $s_{i,k}$ we need to know the value of $f_{i,k}$ at every time step. This can be done easily in a classical computer. However, in a quantum computer, it is difficult to use the nonlinear slope limiter. If we use the Hamiltonian simulation to evolve the solution vector from $t=0$ to $t=T$, we can't get the solution at every time step. If we use the Hamiltonian simulation to evolve the solution vector from $t=(n-1)\Delta t$ to $t=n\Delta t$ for $n=1,\ldots,T/\Delta t$, then we need to measure the quantum state at every time step, which costs a lot.
\end{remark}

\subsection{Schr\"odingerization}\label{sec:interface1Schro}
We then write the above ODE into a homogeneous form:
\begin{equation}\label{equ:interface1homogeneous}
    \frac{\d }{\d t}\begin{pmatrix}
        \boldsymbol{f}(t)\\
        \boldsymbol{r}(t)
    \end{pmatrix}=\begin{pmatrix}
        A & \text{diag}(\boldsymbol{b}(t))/\epsilon\\
        \boldsymbol{0} & \boldsymbol{0}
    \end{pmatrix}\begin{pmatrix}
        \boldsymbol{f}(t)\\
        \boldsymbol{r}(t)
    \end{pmatrix},\quad
    \begin{pmatrix}
        \boldsymbol{f}(0)\\
        \boldsymbol{r}(0)
    \end{pmatrix}=\begin{pmatrix}
        \boldsymbol{f}(0)\\
        \epsilon \boldsymbol{1}
    \end{pmatrix},
\end{equation}
where $\epsilon=\max_t\|\boldsymbol{b}(t)\|_\infty$, $\mathbf{0}$ is a $N_xN_\xi\times N_xN_\xi$ matrix all of $0$ and $\mathbf{1}$ is a $N_xN_\xi$-dimensional vector all of $1$. Then according to \eqref{equ:H1H2}
\begin{equation}
    H_{1}=\frac{1}{2}\begin{pmatrix}
    A+A^\dagger&\mathrm{diag}(\boldsymbol{b}(t))/\epsilon\\
    \mathrm{diag}(\boldsymbol{b}(t)^{\dagger})/\epsilon&\mathbf{0}\end{pmatrix},\quad H_{2}=\frac{1}{2i}\begin{pmatrix}
    A-A^\dagger&\mathrm{diag}(\boldsymbol{b}(t))/\epsilon\\
    -\mathrm{diag}(\boldsymbol{b}(t)^{\dagger})/\epsilon&\mathbf{0}\end{pmatrix}
\end{equation}
and \eqref{equ:interface1homogeneous} can be rewriten as 
\begin{equation}
    \frac{\d}{\d t}\boldsymbol{u}=H_1\boldsymbol{u}+iH_2\boldsymbol{u}.\quad\boldsymbol{u}(0)=\begin{pmatrix}
        \boldsymbol{f}(0)\\
        \epsilon \boldsymbol{1}
    \end{pmatrix}.
\end{equation}
Using the warped phase transformation\eqref{equ:warped} and then taking the Fourier spectral mathod on $p$ in position basis, we can get \eqref{equ:v2w}. After change of variables $\tilde{\bm{w}}(t) = (I \otimes \Phi^{-1})\bm{w}(t)$ and simulating the Hamiltonian system \eqref{equ:generalSchr} form $t=0$ to $t=T$, we can get $\tilde{\boldsymbol{w}}(T)$. Then applying the inverse quantum Fourier transform (IQFT) $\bm{w}(T) = (I \otimes \Phi)\tilde{\bm{w}}(T)$ and using the one point recovery \eqref{equ:restore1point}, we can get the  quantum state of the solution.
The procedure from initial value $\boldsymbol{f}(0)$ to solution $\boldsymbol{f}(T)$ can be summarized as follows:
\begin{multline}\label{equ:procedure}
    \ket{\boldsymbol{f}(0)}\xrightarrow{\text{homogeneous}} \ket{\boldsymbol{u }(0)}\xrightarrow{\text{warped phase}} \ket{\boldsymbol{w}(0)}\xrightarrow{\text{QFT on } p} \ket{\tilde{\boldsymbol{w}}(0)} \xrightarrow{\e^{-iHT}} \ket{\tilde{\boldsymbol{w}}(T)} \\ 
    \xrightarrow{\text{IQFT on } p} \ket{\boldsymbol{w}(T)} \xrightarrow{\text{recovery}} 
    \ket{\boldsymbol{u }(T)} \xrightarrow{\text{measurements}} \ket{\boldsymbol{f }(T)}
\end{multline}

The following theorem give the complexity of Schr\"odingerization based Hamiltonian-preserving scheme combined  in one dimensional case.

\begin{theorem}\label{thm:1dim_complexity1}
    In the 1-dimensional case, given the initial quantum state $\ket{\boldsymbol{u}(0)}$, assume that the $x$ grid number, $\xi$ grid number and extended $p$ grid number are $N_x=2^{n_x}$, $N_\xi=2^{n_\xi}$ and $N_p=2^{n_p}$ respectively, the inhomogeneous term $\boldsymbol{b}$ is independent of time and the error of Hamiltonian-preserving scheme, spectral method in $p$ and Hamiltonian simulation are all $\epsilon$. With the Schr\"odingerization method, the state $\ket{\boldsymbol{u}(t)}$ can be simulated for time $T$, and success probability is at least $1-2\epsilon$ with
    \begin{itemize}
        \item $\mathcal{O}(QT\epsilon^{-2}+\frac{\log\epsilon^{-1}}{\log\log\epsilon^{-1}})$ queries and a factor $\mathcal{O}(3\log(\epsilon^{-1})+\log(\epsilon^{-1})\polylog(\log(\epsilon^{-1})))$ additional quantum gates, when solution has discontinuities only at interfaces,
        \item $\mathcal{O}(QT\epsilon^{-3}+\frac{\log\epsilon^{-1}}{\log\log\epsilon^{-1}})$ queries and a factor $\mathcal{O}(5\log(\epsilon^{-1})+\log(\epsilon^{-1})\polylog(\log(\epsilon^{-1})))$ additional quantum gates, when solution has discontinuities not only at interfaces,
    \end{itemize}
    where $Q:=2\max_{x\in\mathcal{I}}\left\{\left\lceil\frac{c^+(x)}{c^-(x)}\right\rceil,\left\lceil\frac{c^-(x)}{c^+(x)}\right\rceil\right\}+2\ge s_c(A_1)\ge 3$.
\end{theorem}
\begin{proof}
    We use Lemma \ref{lem:complexity} to give the proof. Matrix elements is specified to $\log(\epsilon^{-1})$ bits of precision.
    
    From previous discussion, 
    \begin{equation*}
        \begin{aligned}
            s(H)=s(H_1\otimes D_\mu-H_2\otimes I)&\le \max\{s(A+A^\dagger),s(A-A^\dagger)\}+s(\text{diag}(\boldsymbol{b}(t))/\epsilon)\\
            &\le \max\{s(A_1+A_1^\dagger),s(A_1-A_1^\dagger)\}+\max\{s(A_2+A_2^\dagger),s(A_2-A_2^\dagger)\}-1+1\\
            &\le (s_r(A_1)+s_c(A_1)-1)+(3)\\
            &\le Q+6.
        \end{aligned}
    \end{equation*}
    In the $p$-direction, the initial value is only continuous but not differentiable. According to spectral method, in order to reach precision $\epsilon$, the total number of $p$ grid points is $N_p=\mathcal{O}(\epsilon^{-1})$.
    The total number of points in $x$ and $\xi$-direction is related to where the discontinuities lie.
    \begin{itemize}
        \item When solution has discontinuities only at interfaces,the  Hamiltonian-preserving scheme can give the first order accuracy. Thus $N_x=N_\xi=\mathcal{O}(\epsilon^{-1})$ and $\|H\|_{\max}\le\|H_{1}\otimes D_{p}\|_{\max}+\|H_{2}\|_{\max}\le\mathcal{O}(\left\|H_{1}\right\|_{\max}N_p)=\mathcal{O}((N_x+N_\xi)N_p)=\mathcal{O}(\epsilon^{-2})$. The query complexity is 
        \begin{equation*}
            \mathcal{O}(QT\epsilon^{-2}+\frac{\log\epsilon^{-1}}{\log\log\epsilon^{-1}}).
        \end{equation*} 
        Since the quantum Fourier transformation in Schr\"odingerization can be implemented by using $\mathcal{O}(m\log m)$ gates \cite{namApproximateQuantumFourier2020}, the number of overall additional quantum gates is 
        \begin{equation*}
            \mathcal{O}(3\log(\epsilon^{-1})+\log(\epsilon^{-1})\polylog(\log(\epsilon^{-1}))+\log(\epsilon^{-1})\log(\log(\epsilon^{-1}))).
        \end{equation*}
        \item When solution has discontinuities not only at interfaces but also elsewhere, the Hamiltonian-preserving scheme can give the halfth order precision. Thus $N_x=N_\xi=\mathcal{O}(\epsilon^{-2})$ and $\|H\|_{\max}\le\mathcal{O}(\epsilon^{-3})$. The query complexity is 
        \begin{equation*} 
            \mathcal{O}(QT\epsilon^{-3}+\frac{\log\epsilon^{-1}}{\log\log\epsilon^{-1}}).
        \end{equation*} 
        and the number of overall additional quantum gates is 
        \begin{equation*}
            \mathcal{O}(5\log(\epsilon^{-1})+\log(\epsilon^{-1})\polylog(\log(\epsilon^{-1}))+\log(\epsilon^{-1})\log(\log(\epsilon^{-1}))).
        \end{equation*}
    \end{itemize}
\end{proof}
\begin{remark}
    For the smoother initial data $g(p)\in H^r((L,R))$, $N_p=\mathcal{O}(\epsilon^{-\frac{1}{r}})$.
\end{remark}
\begin{remark}
    In the classical implementation of Hamiltonian-preserving schemes, one can use any time discretization for the time derivative. Here we use the forward Euler time descretization as an example to analyze the classical complexity which contains the number of products and quotients. After discretization, the ODE $\boldsymbol{f}'(t)=A\boldsymbol{f}(t)+\boldsymbol{b}(t)$ becomes $\boldsymbol{f}^{n+1}=(\Delta tA+I)\boldsymbol{f}^n+\Delta t\boldsymbol{b}$ for $n=0,\ldots,\frac{T}{N_t}-1$. To satisfy the CFL condition $\Delta t\max_{i,j}\left[\frac{c_i}{\Delta x}+\frac{\frac{\left|c_{i+\frac{1}{2}}^--c_{i-\frac{1}{2}}^+\right|}{\Delta x}}{\Delta\xi}|\xi_j|\right]\leqslant1$, a time step $\Delta t=\mathcal{O}(\Delta x,\Delta \xi)$ is needed \cite{jinHamiltonianpreservingSchemesLiouville2006}, so we take $N_t=\mathcal{O}(TN_x)$. Then the classical complexity is 
    \begin{equation}
        \mathcal{O}(s_r(\Delta tA+I)N_x N_\xi N_t)=\begin{cases}
            \mathcal{O}(5T\epsilon^{-3}) & \text{when solution has discontinuities only at interfaces},\\
            \mathcal{O}(5T\epsilon^{-6})& \text{when solution has discontinuities not only at interfaces}.
        \end{cases}
    \end{equation}
    There is a polynomial advantage in $\epsilon$ for quantum algorithms.
\end{remark}
\begin{remark}
    Theorem \ref{thm:1dim_complexity1} gives the query complexity and the number of additional quantum gates. The gate complexity of the unitary oracle $O$ that provides a description of $H$ may depend on the matrix dimension, so the actual gate complexity may be larger than query complexity. Also, according to Theorem 4.4 in \cite{jinSchrodingerizationBasedQuantumAlgorithms2025}, the length of truncated interval in $p$ is proportional to the maximum and minimum eigenvalues of $H_1$, so when the $N_x$ or $N_\xi$ doubles, the $N_p$ doubles too.
\end{remark}

\section{Quantum simulation of Liouville equation with interface condition \eqref{equ:ddimLiouvilleOpticsinterfaceconditionPartial} in two space dimension }\label{sec:interface2d}

Consider the Liouville equation in two space dimension:
\begin{equation}\label{equ:2dimLiouvilleOptics}
    f_{t}+\frac{c(x,y)\xi}{\sqrt{\xi^{2}+\eta^{2}}}f_{x}+\frac{c(x,y)\eta}{\sqrt{\xi^{2}+\eta^{2}}}f_{y}-c_{x}\sqrt{\xi^{2}+\eta^{2}}f_{\xi}-c_{y}\sqrt{\xi^{2}+\eta^{2}}f_{\eta}=0.
\end{equation}

\subsection{The original algorithm}
We employ a uniform mesh with grid points at $x_{i+\frac{1}{2}},y_{j+\frac{1}{2}},\xi_{k+\frac{1}{2}},\eta_{l+\frac{1}{2}},i=0,\ldots,N_x,j=0,\ldots,N_y,k=0,\ldots,N_\xi,l=0,\ldots,N_\eta$ in each direction. The cells are centered at $(x_{i},y_{j},\xi_{k},\eta_{l})$ with $x_{i}=\frac{1}{2}(x_{i+\frac{1}{2}}+x_{i-\frac{1}{2}}),y_{j}=\frac{1}{2}(y_{j+\frac{1}{2}}+y_{j-\frac{1}{2}}),\xi_{k}=\frac{1}{2}(\xi_{k+\frac{1}{2}}+\xi_{k-\frac{1}{2}}),\eta_{i}=\frac{1}{2}(\eta_{i+\frac{1}{2}}+\eta_{i-\frac{1}{2}})$. The mesh size is denoted by $\Delta x=x_{i+\frac{1}{2}}-x_{i-\frac{1}{2}},\Delta y=y_{j+\frac{1}{2}}-y_{j-\frac{1}{2}},\Delta\xi=\xi_{k+\frac{1}{2}}-\xi_{k-\frac{1}{2}},\Delta\eta=\eta_{l+\frac{1}{2}}-\eta_{l-\frac{1}{2}}.$ We define the cell average of $f$ as
\begin{equation*}
    f_{ijkl}=\frac{1}{\Delta x\Delta y\Delta\xi\Delta\eta}\int_{x_{i-\frac{1}{2}}}^{x_{i+\frac{1}{2}}}\int_{y_{j-\frac{1}{2}}}^{y_{j+\frac{1}{2}}}\int_{\xi_{k-\frac{1}{2}}}^{\xi_{k+\frac{1}{2}}}\int_{\eta_{l-\frac{1}{2}}}^{\eta_{l+\frac{1}{2}}}f(x,y,\xi,\eta,t)\mathrm{d}\eta\mathrm{d}\xi\mathrm{d}y\mathrm{d}x.
\end{equation*}
Similar to the 1D case, we approximate $c(x, y)$ by a piecewise bilinear function, and for convenience, we always provide two interface values of $c$ at each cell interface. When $c$ is smooth at a cell interface, the two interface values are identical. We also define the averaged wave speed in a cell by averaging the four cell interface values
\begin{equation*}
    c_{ij}=\frac{c_{i-\frac{1}{2},j}^{+}+c_{i+\frac{1}{2},j}^{-}+c_{i,j-\frac{1}{2}}^{+}+c_{i,j+\frac{1}{2}}^{-}}{4}.
\end{equation*}

The 2D Liouville equation \eqref{equ:2dimLiouvilleOptics} can be semi-discretized as
\begin{multline}\label{equ:2dimLiouvilleOpticsSmeidis}
    \left(f_{ijkl}\right)_{t}+\frac{c_{ij}\xi_{k}}{\Delta x\sqrt{\xi_{k}^{2}+\eta_{l}^{2}}}\left(f_{i+\frac{1}{2},jkl}^{-}-f_{i-\frac{1}{2},jkl}^{+}\right)+\frac{c_{ij}\eta_{l}}{\Delta y\sqrt{\xi_{k}^{2}+\eta_{l}^{2}}}\left(f_{i,j+\frac{1}{2},kl}^{-}-f_{i,j-\frac{1}{2},kl}^{+}\right)\\
    -\frac{c_{i+\frac{1}{2},j}^{-}-c_{i-\frac{1}{2},j}^{+}}{\Delta x\Delta\xi}\sqrt{\xi_{k}^{2}+\eta_{l}^{2}}\left(f_{ij,k+\frac{1}{2},l}-f_{ij,k-\frac{1}{2},l}\right)-\frac{c_{i,j+\frac{1}{2}}^{-}-c_{i,j-\frac{1}{2}}^{+}}{\Delta y\Delta\eta}\sqrt{\xi_{k}^{2}+\eta_{l}^{2}}\left(f_{ijk,l+\frac{1}{2}}-f_{ijk,l-\frac{1}{2}}\right)=0
\end{multline}
where the numerical fluxes $f_{ij,k+\frac{1}{2},l},f_{ijk,l+\frac{1}{2}}$ are provided by the upwind approximation, and the split fluxes values
$f_{i+\frac{1}{2},jkl}^{-},f_{i-\frac{1}{2},jkl}^{+},f_{i,j+\frac{1}{2},kl}^{-},f_{i,j-\frac{1}{2},kl}^{+}$ should be obtained using similar but slightly different algorithm for the 1D case, since the particle behavior at the interface is different in 2D from that in 1D.
For example, to evaluate $f_{i+\frac{1}{2},jkl}^{\pm}$ we can extend Algorithm \ref{alg:fluxOptics1tranPartial} as Algorithm \ref{alg:fluxOptics1tranPartial2d}.

The flux $f_{i,j+\frac{1}{2},kl}^{\pm}$ can be constructed similarly.

\begin{algorithm}[!htb]
  \label{alg:fluxOptics1tranPartial2d}
  \caption{Computation of the numerical flux in \eqref{equ:2dimLiouvilleOpticsSmeidis}}
  \SetAlgoLined
  \KwIn{$\xi_k$, $\eta_l$, $c^-_{i+\frac{1}{2},j}$, $c^+_{i+\frac{1}{2},j}$, $\{f_{ijkl},k=1,\ldots,N_\xi\}$ and $\{f_{i+1,jkl},k=1,\ldots,N_\xi\}$  }
  \KwOut{$f^-_{i+\frac{1}{2},jkl}$ and $f^+_{i+\frac{1}{2},jkl}$}
  \If{$\xi_k>0$}{
      $f^-_{i+\frac{1}{2},jkl}=f_{ijkl}$,\quad $\xi_{k_1}=-\xi_k$\;
      \eIf{$\left(\frac{c_{i+\frac{1}{2},j}^{+}}{c_{i+\frac{1}{2},j}^{-}}\right)^{2}(\xi_{k})^{2}+\left[\left(\frac{c_{i+\frac{1}{2},j}^{+}}{c_{i+\frac{1}{2},j}^{-}}\right)^{2}-1\right](\eta_{l})^{2}>0$}
      {$\xi^{-}=\sqrt{\left(\frac{c_{i+\frac{1}{2},j}^{+}}{c_{i+\frac{1}{2},j}^{-}}\right)^{2}\left(\xi_{k}\right)^{2}+\left[\left(\frac{c_{i+\frac{1}{2},j}^{+}}{c_{i+\frac{1}{2},j}^{-}}\right)^{2}-1\right]\left(\eta_{l}\right)^{2}}$\;
      \If{$\xi_{k'}\le\xi^-<\xi_{k'+1}$ for some $k'$}
      {$\gamma^{+}_{kl}=\frac{\xi_{k}}{\sqrt{\left(\xi_{k}\right)^{2}+\left(\eta_{l}\right)^{2}}},\quad\gamma^{-}_{kl}=\frac{\xi^{-}}{\sqrt{\left(\xi^-\right)^{2}+\left(\eta_{l}\right)^{2}}}$\;
      $\alpha^{R}_{i+\frac{1}{2},jkl}=\left(\frac{c_{i+\frac{1}{2},j}^{+}\gamma^{-}_{kl}-c_{i+\frac{1}{2},j}^{-}\gamma^{+}_{kl}}{c_{i+\frac{1}{2},j}^{+}\gamma^{-}_{kl}+c_{i+\frac{1}{2},j}^{-}\gamma^{+}_{kl}}\right)^{2},\quad\alpha^{T}_{i+\frac{1}{2},jkl}=1-\alpha^{R}_{i+\frac{1}{2},jkl}$\;
      $f_{i+\frac{1}{2},jkl}^{+}=\alpha^{T}_{i+\frac{1}{2},jkl}\left(\frac{\xi_{k^{\prime}+1}-\xi^{-}}{\Delta\xi}f_{ij,k^{\prime},l}+\frac{\xi^{-}-\xi_{k^{\prime}}}{\Delta\xi}f_{ij,k^{\prime}+1,l}\right)+\alpha^{R}_{i+\frac{1}{2},jkl}f_{i+1,j,k_{1},l}$\;
      }
      }
      {$f_{i+\frac{1}{2},jkl}^{+}=f_{i+1,j,k_{1},l}$\;}
      }
  \If{$\xi_k<0$}{
      $f_{i+\frac{1}{2},jkl}^{+}=f_{i+1,jkl},\quad\xi_{k_{1}}=-\xi_{k}$\;
      \eIf{$\left(\frac{c_{i+\frac{1}{2},j}^{-}}{c_{i+\frac{1}{2},j}^{+}}\right)^{2}\left(\xi_{k}\right)^{2}+\left[\left(\frac{c_{i+\frac{1}{2},j}^{-}}{c_{i+\frac{1}{2},j}^{+}}\right)^{2}-1\right]\left(\eta_{l}\right)^{2}>0$}
      {$\xi^{+}=-\sqrt{\left(\frac{c_{i+\frac{1}{2},j}^{-}}{c_{i+\frac{1}{2},j}^{+}}\right)^{2}\left(\xi_{k}\right)^{2}+\left[\left(\frac{c_{i+\frac{1}{2},j}^{-}}{c_{i+\frac{1}{2},j}^{+}}\right)^{2}-1\right]\left(\eta_{l}\right)^{2}}$\;
      \If{$\xi_{k'}\le\xi^+<\xi_{k'+1}$ for some $k'$}
      {$\gamma^{+}_{kl}=\frac{|\xi^{+}|}{\sqrt{\left(\xi^{+}\right)^{2}+\left(\eta_{l}\right)^{2}}},\quad\gamma^{-}_{kl}=\frac{|\xi_{k}|}{\sqrt{\left(\xi_{k}\right)^{2}+\left(\eta_{l}\right)^{2}}}$\;
      $\alpha^{R}_{i+\frac{1}{2},jkl}=\left(\frac{c_{i+\frac{1}{2},j}^{+}\gamma^{-}_{kl}-c_{i+\frac{1}{2},j}^{-}\gamma^{+}_{kl}}{c_{i+\frac{1}{2},j}^{+}\gamma^{-}_{kl}+c_{i+\frac{1}{2},j}^{-}\gamma^{+}_{kl}}\right)^{2},\quad\alpha^{T}_{i+\frac{1}{2},jkl}=1-\alpha^{R}_{i+\frac{1}{2},jkl}$\;
      $f_{i+\frac{1}{2},jkl}^{-}=\alpha^{T}_{i+\frac{1}{2},jkl}\left(\frac{\xi_{k^{\prime}+1}-\xi^{+}}{\Delta\xi}f_{i+1,j,k^{\prime},l}+\frac{\xi^{+}-\xi_{k^{\prime}}}{\Delta\xi}f_{i+1,j,k^{\prime}+1,l}\right)+\alpha^{R}_{i+\frac{1}{2},jkl}f_{ij,k_{1},l}$\;
      }
      }
      {$f_{i+\frac{1}{2},jkl}^{-}=f_{i,j,k_{1},l}$\;}
      }
\end{algorithm}

\subsection{The ODE form}\label{sec:interface2dODEform}
In the following, we convert the semi-discrete scheme \eqref{equ:2dimLiouvilleOpticsSmeidis} and Algorithm \ref{alg:fluxOptics1tranPartial2d} into the ODE form \eqref{equ:ODElinear}. Assume that $N_x$, $N_y$, $N_\xi$ and $N_\eta$ are even numbers and the truncated interval in $x$, $y$, $\xi$ and $\eta$-direction are symmetric about $0$. (For asymmetrical case, the matrix $A$ in ODE \eqref{equ:ODElinear} will be slightly modified.) We also assume that interfaces don't appear at $x_\frac{1}{2}$, $x_{N_x+\frac12}$, $y_\frac{1}{2}$ and $y_{N_y+\frac12}$. Denote $\boldsymbol{f}(t)=\sum_{i=1}^{N_x}\sum_{j=1}^{N_y}\sum_{k=1}^{N_\xi}\sum_{l=1}^{N_\eta}f_{ijkl}(t)\ket{i-1}\ket{j-1}\ket{k-1}\ket{l-1}$. This definition specifies not only the order of the equations but also the order of the variables. We can also arrange $\boldsymbol{f}$ in different orders, then we only need to change the corresponding columns and rows of $A$ to get the new matrix.

If $\xi_k>0$, for the term $\frac{c_{ij}\xi_{k}}{\Delta x\sqrt{\xi_{k}^{2}+\eta_{l}^{2}}}\left(f_{i+\frac{1}{2},jkl}^{-}-f_{i-\frac{1}{2},jkl}^{+}\right)$ in \eqref{equ:2dimLiouvilleOpticsSmeidis}, $f_{i+\frac{1}{2},jkl}^{-}=f_{ijkl}$ and $f^+_{i-\frac12,jkl}$ depends on the sign of $\left(\frac{c_{i-\frac{1}{2},j}^{+}}{c_{i-\frac{1}{2},j}^{-}}\right)^{2}(\xi_{k})^{2}+\left[\left(\frac{c_{i-\frac{1}{2},j}^{+}}{c_{i-\frac{1}{2},j}^{-}}\right)^{2}-1\right](\eta_{l})^{2}$. We define two step funtions
\begin{equation}\label{equ:step}
    g_{i+\frac{1}{2},jkl}^T(z)=\alpha^{T}_{i+\frac{1}{2},jkl}\chi_{(0,\infty)}(z),\quad 
    g_{i+\frac{1}{2},jkl}^R(z)=\chi_{(-\infty,0]}(z)+\alpha^{R}_{i+\frac{1}{2},jkl}\chi_{(0,\infty)}(z),
\end{equation}
where $\chi_{I}$ is the indicator function of set $I$
\begin{equation}
    \chi_{I}\left(z\right)=\left\{\begin{array}{ll}1,&\mathrm{if~}z\in I,\\0,&\mathrm{if~}z\not\in I.\end{array}\right.
\end{equation}
The above two step funcitons control the reflection and transmission at the interface.
Define
\begin{gather}
    a^T_{i+\frac{1}{2},jkl}=g_{i+\frac{1}{2},jkl}^T\left(\left(\frac{c_{i+\frac{1}{2},j}^{+}}{c_{i+\frac{1}{2},j}^{-}}\right)^{2}(\xi_{k})^{2}+\left[\left(\frac{c_{i+\frac{1}{2},j}^{+}}{c_{i+\frac{1}{2},j}^{-}}\right)^{2}-1\right](\eta_{l})^{2}\right),\\
    a^R_{i+\frac{1}{2},jkl}=g_{i+\frac{1}{2},jkl}^R\left(\left(\frac{c_{i+\frac{1}{2},j}^{+}}{c_{i+\frac{1}{2},j}^{-}}\right)^{2}(\xi_{k})^{2}+\left[\left(\frac{c_{i+\frac{1}{2},j}^{+}}{c_{i+\frac{1}{2},j}^{-}}\right)^{2}-1\right](\eta_{l})^{2}\right).
\end{gather}
The search of which interval $\xi^-$ lies in involves a nonlinear operation. We need to use hat function \eqref{equ:hat} to encode this nonlinearity into the matrix $A$. 
Then 
\begin{eqnarray}\nonumber
f_{i-\frac{1}{2},jkl}^{+}=&&\sum_{k'=1}^{N_\xi}a^{T}_{i-\frac{1}{2},jkl}h\left(\xi^--\xi_{k'}\right)f_{i-1,jk'l}+a^{R}_{i-\frac{1}{2},jkl}f_{i,j,k_{1},l}\\
=&&
\sum_{k'=1}^{N_\xi}a^{T}_{i-\frac{1}{2},jkl}h\left(\sqrt{\left(\frac{c_{i-\frac{1}{2},j}^{+}}{c_{i-\frac{1}{2},j}^{-}}\right)^{2}(\xi_{k})^{2}+\left[\left(\frac{c_{i-\frac{1}{2},j}^{+}}{c_{i-\frac{1}{2},j}^{-}}\right)^{2}-1\right](\eta_{l})^{2}}-\xi_{k'}\right)f_{i-1,jk'l}\nonumber\\
&&+a^{R}_{i-\frac{1}{2},jkl}f_{i,j,k_{1},l}
\end{eqnarray}
and 
\begin{multline*}
    \frac{c_{ij}\xi_{k}}{\Delta x\sqrt{\xi_{k}^{2}+\eta_{l}^{2}}}\left(f_{i+\frac{1}{2},jkl}^{-}-f_{i-\frac{1}{2},jkl}^{+}\right)=
    \frac{c_{ij}\xi_{k}}{\Delta x\sqrt{\xi_{k}^{2}+\eta_{l}^{2}}}f_{ijkl}\\
    -\frac{c_{ij}\xi_{k}}{\Delta x\sqrt{\xi_{k}^{2}+\eta_{l}^{2}}}\sum_{k'=1}^{N_\xi}a^{T}_{i-\frac{1}{2},jkl}h\left(\sqrt{\left(\frac{c_{i-\frac{1}{2},j}^{+}}{c_{i-\frac{1}{2},j}^{-}}\right)^{2}(\xi_{k})^{2}+\left[\left(\frac{c_{i-\frac{1}{2},j}^{+}}{c_{i-\frac{1}{2},j}^{-}}\right)^{2}-1\right](\eta_{l})^{2}}-\xi_{k'}\right)f_{i-1,jk'l}\\ 
    -\frac{c_{ij}\xi_{k}}{\Delta x\sqrt{\xi_{k}^{2}+\eta_{l}^{2}}}a^{R}_{i-\frac{1}{2},jkl}f_{i,j,k_{1},l}.
\end{multline*}

If $\xi_k<0$, for the term $\frac{c_{ij}\xi_{k}}{\Delta x\sqrt{\xi_{k}^{2}+\eta_{l}^{2}}}\left(f_{i+\frac{1}{2},jkl}^{-}-f_{i-\frac{1}{2},jkl}^{+}\right)$ in \eqref{equ:2dimLiouvilleOpticsSmeidis}, similarly, define
\begin{gather}
    a^T_{i+\frac{1}{2},jkl}=g_{i+\frac{1}{2},jkl}^T\left(\left(\frac{c_{i+\frac{1}{2},j}^{-}}{c_{i+\frac{1}{2},j}^{+}}\right)^{2}(\xi_{k})^{2}+\left[\left(\frac{c_{i+\frac{1}{2},j}^{-}}{c_{i+\frac{1}{2},j}^{+}}\right)^{2}-1\right](\eta_{l})^{2}\right),\\
    a^R_{i+\frac{1}{2},jkl}=g_{i+\frac{1}{2},jkl}^R\left(\left(\frac{c_{i+\frac{1}{2},j}^{-}}{c_{i+\frac{1}{2},j}^{+}}\right)^{2}(\xi_{k})^{2}+\left[\left(\frac{c_{i+\frac{1}{2},j}^{-}}{c_{i+\frac{1}{2},j}^{+}}\right)^{2}-1\right](\eta_{l})^{2}\right),
\end{gather}
then we can get 
\begin{multline*}
    \frac{c_{ij}\xi_{k}}{\Delta x\sqrt{\xi_{k}^{2}+\eta_{l}^{2}}}\left(f_{i+\frac{1}{2},jkl}^{-}-f_{i-\frac{1}{2},jkl}^{+}\right)=\\
    \frac{c_{ij}\xi_{k}}{\Delta x\sqrt{\xi_{k}^{2}+\eta_{l}^{2}}}\sum_{k'=1}^{N_\xi}a^{T}_{i+\frac{1}{2},jkl}h\left(-\sqrt{\left(\frac{c_{i+\frac{1}{2},j}^{-}}{c_{i+\frac{1}{2},j}^{+}}\right)^{2}(\xi_{k})^{2}+\left[\left(\frac{c_{i+\frac{1}{2},j}^{-}}{c_{i+\frac{1}{2},j}^{+}}\right)^{2}-1\right](\eta_{l})^{2}}-\xi_{k'}\right)f_{i+1,jk'l}\\ 
    +\frac{c_{ij}\xi_{k}}{\Delta x\sqrt{\xi_{k}^{2}+\eta_{l}^{2}}}a^{R}_{i+\frac{1}{2},jkl}f_{i,j,k_{1},l}
    -\frac{c_{ij}\xi_{k}}{\Delta x\sqrt{\xi_{k}^{2}+\eta_{l}^{2}}}f_{ijkl}.
\end{multline*}

Define 
\begin{equation}\label{equ:betaxi}
    \beta_{i+\frac12,j,k,l,k'}^\xi=\left\{\begin{array}{ll}
    h\left(\sqrt{\left(\frac{c_{i+\frac{1}{2},j}^{+}}{c_{i+\frac{1}{2},j}^{-}}\right)^{2}(\xi_{k})^{2}+\left[\left(\frac{c_{i+\frac{1}{2},j}^{+}}{c_{i+\frac{1}{2},j}^{-}}\right)^{2}-1\right](\eta_{l})^{2}}-\xi_{k'}\right),     \\
    i=0,\ldots,N_x,\quad j=1,\ldots,N_y,\quad k=\frac{N_\xi}{2}+1,\ldots,N_\xi,\quad l=1,\ldots,N_\eta,\quad k'=1,\ldots,N_\xi, \\
    h\left(-\sqrt{\left(\frac{c_{i+\frac{1}{2},j}^{-}}{c_{i+\frac{1}{2},j}^{+}}\right)^{2}(\xi_{k})^{2}+\left[\left(\frac{c_{i+\frac{1}{2},j}^{-}}{c_{i+\frac{1}{2},j}^{+}}\right)^{2}-1\right](\eta_{l})^{2}}-\xi_{k'}\right),     \\
    i=0,\ldots,N_x,\quad j=1,\ldots,N_y,\quad k=1,\ldots, \frac{N_\xi}{2},\quad l=1,\ldots,N_\eta,\quad k'=1,\ldots,N_\xi. 
    \end{array}\right.
\end{equation}
Here we use $i+\frac12$ to show that $\beta_{i+\frac12,j,k,l,k'}^\xi$ is related to the coefficients $c^-_{i+\frac{1}{2},j}$ and $c^+_{i+\frac{1}{2},j}$ of the interface at $x_{i+\frac{1}{2},j}$.
Then
\begin{multline*}
    \frac{c_{ij}\xi_{k}}{\Delta x\sqrt{\xi_{k}^{2}+\eta_{l}^{2}}}\left(f_{i+\frac{1}{2},jkl}^{-}-f_{i-\frac{1}{2},jkl}^{+}\right)=\\
    \left\{\begin{array}{l}
        \frac{c_{ij}\xi_{k}}{\Delta x\sqrt{\xi_{k}^{2}+\eta_{l}^{2}}}\left(f_{ijkl}-\sum_{k'=1}^{N_\xi}a_{i-\frac12,jkl}^T\beta_{i-\frac12,j,k,l,k'}^\xi f_{i-1,jk'l}-a_{i-\frac12,jkl}^R f_{i,j,k_1,l}\right),     \\ 
        \qquad i=1,\ldots,N_x,\quad j=1,\ldots,N_y,\quad k=\frac{N_\xi}{2}+1,\ldots,N_\xi,\quad l=1,\ldots,N_\eta, \\
        \frac{c_{ij}\xi_{k}}{\Delta x\sqrt{\xi_{k}^{2}+\eta_{l}^{2}}}\left(\sum_{k'=1}^{N_\xi}a^{T}_{i+\frac{1}{2},jkl}\beta_{i+\frac12,j,k,l,k'}^\xi f_{i+1,jk'l}
        +a^{R}_{i+\frac{1}{2},jkl}f_{i,j,k_{1},l}-f_{ijkl}\right),    \\ 
        \qquad i=1,\ldots,N_x,\quad j=1,\ldots,N_y,\quad k=1,\ldots,\frac{N_\xi}{2},\quad l=1,\ldots,N_\eta.
    \end{array}
    \right.
\end{multline*}
When $c^-_{i+\frac{1}{2},j}=c^+_{i+\frac{1}{2},j}$ the above scheme degenerates to the upwind scheme. 
Since we have assumed that interfaces don't appear at $x_\frac{1}{2}$ and $x_{N_x+\frac12}$, we can use the boundary condition to give the value of $f_{0jkl}(t)$ and $f_{N_x+1,jkl}(t)$.
When $\xi_k>0$ (resp. $\xi_k<0$), we impose inflow boundary condition at $(x_0,y_j,\xi_k,\eta_l)$ (resp. $(x_{N_x+1},y_j,\xi_k,\eta_l)$) and outflow boundary condition at $(x_{N_x+1},y_j,\xi_k,\eta_l)$ (resp. $(x_0,y_j,\xi_k,\eta_l)$).
Write the above equations in the form $A_1\boldsymbol{f}+b_1$ of a product of a matrix and a vector with an inhomogeneous term representing boundary condition, where the equations and variables are in order of $\boldsymbol{f}$, we can get a $N_xN_yN_\xi N_\eta\times N_xN_yN_\xi N_\eta$ matrix
\begin{multline}\label{equ:mtx1}
    A_1=\sum_{i=1}^{N_x}\sum_{j=1}^{N_y}\sum_{k=\frac{N_\xi}{2}+1}^{N_\xi}\sum_{l=1}^{N_\eta}\frac{c_{ij}\xi_{k}}{\Delta x\sqrt{\xi_{k}^{2}+\eta_{l}^{2}}}\bra{i-1}\bra{j-1}\bra{k-1}\bra{l-1}\ket{i-1}\ket{j-1}\ket{k-1}\ket{l-1}\\
    -\sum_{i=2}^{N_x}\sum_{j=1}^{N_y}\sum_{k=\frac{N_\xi}{2}+1}^{N_\xi}\sum_{l=1}^{N_\eta}\sum_{k'=1}^{N_\xi}\frac{c_{ij}\xi_{k}}{\Delta x\sqrt{\xi_{k}^{2}+\eta_{l}^{2}}}a_{i-\frac12,jkl}^T\beta_{i-\frac12,j,k,l,k'}^\xi\bra{i-1}\bra{j-1}\bra{k-1}\bra{l-1}\ket{i-2}\ket{j-1}\ket{k'-1}\ket{l-1}\\
    -\sum_{i=2}^{N_x}\sum_{j=1}^{N_y}\sum_{k=\frac{N_\xi}{2}+1}^{N_\xi}\sum_{l=1}^{N_\eta}\frac{c_{ij}\xi_{k}}{\Delta x\sqrt{\xi_{k}^{2}+\eta_{l}^{2}}}a_{i-\frac12,jkl}^R\bra{i-1}\bra{j-1}\bra{k-1}\bra{l-1}\ket{i-1}\ket{j-1}\ket{N_\xi-k}\ket{l-1}\\
    +\sum_{i=1}^{N_x-1}\sum_{j=1}^{N_y}\sum_{k=1}^{\frac{N_\xi}{2}}\sum_{l=1}^{N_\eta}\sum_{k'=1}^{N_\xi}\frac{c_{ij}\xi_{k}}{\Delta x\sqrt{\xi_{k}^{2}+\eta_{l}^{2}}}a^{T}_{i+\frac{1}{2},jkl}\beta_{i+\frac12,j,k,l,k'}^\xi\bra{i-1}\bra{j-1}\bra{k-1}\bra{l-1}\ket{i}\ket{j-1}\ket{k'-1}\ket{l-1}\\
    +\sum_{i=1}^{N_x-1}\sum_{j=1}^{N_y}\sum_{k=1}^{\frac{N_\xi}{2}}\sum_{l=1}^{N_\eta}\frac{c_{ij}\xi_{k}}{\Delta x\sqrt{\xi_{k}^{2}+\eta_{l}^{2}}}a^{R}_{i+\frac{1}{2},jkl}\bra{i-1}\bra{j-1}\bra{k-1}\bra{l-1}\ket{i-1}\ket{j-1}\ket{N_\xi-k}\ket{l-1}\\
    -\sum_{i=1}^{N_x}\sum_{j=1}^{N_y}\sum_{k=1}^{\frac{N_\xi}{2}}\sum_{l=1}^{N_\eta}\frac{c_{ij}\xi_{k}}{\Delta x\sqrt{\xi_{k}^{2}+\eta_{l}^{2}}}\bra{i-1}\bra{j-1}\bra{k-1}\bra{l-1}\ket{i-1}\ket{j-1}\ket{k-1}\ket{l-1},
\end{multline}
and a $N_xN_yN_\xi N_\eta$-dimensional vector 
\begin{multline}\label{equ:inhomogeneoustermx1}
    b_1=-\sum_{j=1}^{N_y}\sum_{k=\frac{N_\xi}{2}+1}^{N_\xi}\sum_{l=1}^{N_\eta}\frac{c_{1j}\xi_{k}}{\Delta x\sqrt{\xi_{k}^{2}+\eta_{l}^{2}}}f_{0jkl}(t)\ket{0}\ket{j-1}\ket{k-1}\ket{l-1}\\
    +\sum_{j=1}^{N_y}\sum_{j=1}^{\frac{N_\xi}{2}}\sum_{l=1}^{N_\eta}\frac{c_{N_x,j}\xi_{k}}{\Delta x\sqrt{\xi_{k}^{2}+\eta_{l}^{2}}}f_{N_x+1, jkl}(t)\ket{N_x-1}\ket{j-1}\ket{k-1}\ket{l-1}.
\end{multline}
To facilitate subsequent complexity analysis, we discuss the sparsity of $A_1$. The sparsity in row is easy to compute $s_r(A_1)\le 4$. Denote the set of interfaces vertical to the $x$-axis as $\mathcal{I}_x=\{(x_{i+\frac{1}{2}},y_j):c(x,y)\text{ is discontinuous on }(x_{i+\frac{1}{2}},y_j)\}$. Then $s_c(A_1)\le2\max_{(x,y)\in\mathcal{I}_x}\left\{\left\lceil\frac{c^+(x,y)}{c^-(x,y)}\right\rceil,\left\lceil\frac{c^-(x,y)}{c^+(x,y)}\right\rceil\right\}+2:=Q_1$.

Similarly, for the term $\frac{c_{ij}\eta_{l}}{\Delta y\sqrt{\xi_{k}^{2}+\eta_{l}^{2}}}\left(f_{i,j+\frac{1}{2},kl}^{-}-f_{i,j-\frac{1}{2},kl}^{+}\right)$ in \eqref{equ:2dimLiouvilleOpticsSmeidis}, define 
\begin{equation}\label{equ:step2}
    g_{i,j+\frac{1}{2},kl}^T(z)=\alpha^{T}_{i,j+\frac{1}{2},kl}\chi_{(0,\infty)}(z),\quad 
    g_{i,j+\frac{1}{2},kl}^R(z)=\chi_{(-\infty,0]}(z)+\alpha^{R}_{i,j+\frac{1}{2},kl}\chi_{(0,\infty)}(z),
\end{equation}
for $\eta_l>0$,
\begin{gather}
    a^T_{i,j+\frac{1}{2},kl}=g_{i,j+\frac{1}{2},kl}^T\left(\left(\frac{c_{i,j+\frac{1}{2}}^{+}}{c_{i,j+\frac{1}{2}}^{-}}\right)^{2}(\eta_{l})^{2}+\left[\left(\frac{c_{i,j+\frac{1}{2}}^{+}}{c_{i,j+\frac{1}{2}}^{-}}\right)^{2}-1\right](\xi_{k})^{2}\right),\\
    a^R_{i,j+\frac{1}{2},kl}=g_{i,j+\frac{1}{2},kl}^R\left(\left(\frac{c_{i,j+\frac{1}{2}}^{+}}{c_{i,j+\frac{1}{2}}^{-}}\right)^{2}(\eta_{l})^{2}+\left[\left(\frac{c_{i,j+\frac{1}{2}}^{+}}{c_{i,j+\frac{1}{2}}^{-}}\right)^{2}-1\right](\xi_{k})^{2}\right),
\end{gather}
for $\eta_l<0$,
\begin{gather}
    a^T_{i,j+\frac{1}{2},kl}=g_{i,j+\frac{1}{2},kl}^T\left(\left(\frac{c_{i,j+\frac{1}{2}}^{-}}{c_{i,j+\frac{1}{2}}^{+}}\right)^{2}(\eta_{l})^{2}+\left[\left(\frac{c_{i,j+\frac{1}{2}}^{-}}{c_{i,j+\frac{1}{2}}^{+}}\right)^{2}-1\right](\xi_{k})^{2}\right),\\
    a^R_{i,j+\frac{1}{2},kl}=g_{i,j+\frac{1}{2},kl}^R\left(\left(\frac{c_{i,j+\frac{1}{2}}^{-}}{c_{i,j+\frac{1}{2}}^{+}}\right)^{2}(\eta_{l})^{2}+\left[\left(\frac{c_{i,j+\frac{1}{2}}^{-}}{c_{i,j+\frac{1}{2}}^{+}}\right)^{2}-1\right](\xi_{k})^{2}\right),
\end{gather}
and
\begin{equation}\label{equ:betaeta}
    \beta_{i,j+\frac12,k,l,l'}^\eta=\left\{\begin{array}{ll}
    h\left(\sqrt{\left(\frac{c_{i,j+\frac{1}{2}}^{+}}{c_{i,j+\frac{1}{2}}^{-}}\right)^{2}(\eta_{l})^{2}+\left[\left(\frac{c_{i,j+\frac{1}{2}}^{+}}{c_{i,j+\frac{1}{2}}^{-}}\right)^{2}-1\right](\xi_{k})^{2}}-\eta_{l'}\right),     \\
    i=1,\ldots,N_x,\quad j=0,\ldots,N_y,\quad k=1,\ldots,N_\xi,\quad l=\frac{N_\eta}{2}+1,\ldots,N_\eta,\quad l'=1,\ldots,N_\eta, \\
    h\left(-\sqrt{\left(\frac{c_{i,j+\frac{1}{2}}^{-}}{c_{i,j+\frac{1}{2}}^{+}}\right)^{2}(\eta_{l})^{2}+\left[\left(\frac{c_{i,j+\frac{1}{2}}^{-}}{c_{i,j+\frac{1}{2}}^{+}}\right)^{2}-1\right](\xi_{k})^{2}}-\eta_{l'}\right),     \\
    i=1,\ldots,N_x,\quad j=0,\ldots,N_y,\quad k=1,\ldots, N_\eta,\quad l=1,\ldots,\frac{N_\eta}{2},\quad l'=1,\ldots,N_\eta. 
    \end{array}\right.
\end{equation}
Then
\begin{multline*}
    \frac{c_{ij}\eta_{l}}{\Delta y\sqrt{\xi_{k}^{2}+\eta_{l}^{2}}}\left(f_{i,j+\frac{1}{2},kl}^{-}-f_{i,j-\frac{1}{2},kl}^{+}\right)=\\
    \left\{\begin{array}{l}
        \frac{c_{ij}\eta_{l}}{\Delta y\sqrt{\xi_{k}^{2}+\eta_{l}^{2}}}\left(f_{ijkl}-\sum_{l'=1}^{N_\eta}a_{i,j-\frac12,kl}^T\beta_{i,j-\frac12,k,l,l'}^\eta f_{i,j-1,k,l'}-a_{i,j-\frac12,kl}^R f_{i,j,k,l_1}\right),     \\ 
        \qquad i=1,\ldots,N_x,\quad j=1,\ldots,N_y,\quad k=1,\ldots,N_\xi,\quad l=\frac{N_\eta}{2}+1,\ldots,N_\eta, \\
        \frac{c_{ij}\eta_{l}}{\Delta y\sqrt{\xi_{k}^{2}+\eta_{l}^{2}}}\left(\sum_{l'=1}^{N_\eta}a^{T}_{i,j+\frac{1}{2},kl}\beta_{i,j+\frac12,k,l,l'}^\eta f_{i,j+1,k,l'}
        +a^{R}_{i,j+\frac{1}{2},kl}f_{i,j,k,l_{1}}-f_{ijkl}\right),    \\ 
        \qquad i=1,\ldots,N_x,\quad j=1,\ldots,N_y,\quad k=1,\ldots,N_\xi,\quad l=1,\ldots,\frac{N_\eta}{2}.
    \end{array}
    \right.
\end{multline*}
Write the above equations in the form $A_2\boldsymbol{f}+b_2$ of a product of a matrix and a vector with an inhomogeneous term representing boundary condition, where the equations and variables are in order of $\boldsymbol{f}$, we can get a $N_xN_yN_\xi N_\eta\times N_xN_yN_\xi N_\eta$ matrix
\begin{multline}\label{equ:mtx2}
    A_2=\sum_{i=1}^{N_x}\sum_{j=1}^{N_y}\sum_{k=1}^{N_\xi}\sum_{l=\frac{N_\eta}{2}+1}^{N_\eta}\frac{c_{ij}\eta_{l}}{\Delta y\sqrt{\xi_{k}^{2}+\eta_{l}^{2}}}\bra{i-1}\bra{j-1}\bra{k-1}\bra{l-1}\ket{i-1}\ket{j-1}\ket{k-1}\ket{l-1}\\
    -\sum_{i=1}^{N_x}\sum_{j=2}^{N_y}\sum_{k=1}^{N_\xi}\sum_{l=\frac{N_\eta}{2}+1}^{N_\eta}\sum_{l'=1}^{N_\eta}\frac{c_{ij}\eta_{l}}{\Delta y\sqrt{\xi_{k}^{2}+\eta_{l}^{2}}}a_{i,j-\frac12,kl}^T\beta_{i,j-\frac12,k,l,l'}^\eta\bra{i-1}\bra{j-1}\bra{k-1}\bra{l-1}\ket{i-1}\ket{j-2}\ket{k-1}\ket{l'-1}\\
    -\sum_{i=1}^{N_x}\sum_{j=2}^{N_y}\sum_{k=1}^{N_\xi}\sum_{l=\frac{N_\eta}{2}+1}^{N_\eta}\frac{c_{ij}\eta_{k}}{\Delta y\sqrt{\xi_{k}^{2}+\eta_{l}^{2}}}a_{i,j-\frac12,kl}^R\bra{i-1}\bra{j-1}\bra{k-1}\bra{l-1}\ket{i-1}\ket{j-1}\ket{k-1}\ket{N_\eta-l}\\
    +\sum_{i=1}^{N_x}\sum_{j=1}^{N_y-1}\sum_{k=1}^{N_\xi}\sum_{l=1}^{\frac{N_\eta}{2}}\sum_{l'=1}^{N_\eta}\frac{c_{ij}\eta_{l}}{\Delta y\sqrt{\xi_{k}^{2}+\eta_{l}^{2}}}a^{T}_{i,j+\frac{1}{2},kl}\beta_{i,j+\frac12,k,l,l'}^\eta\bra{i-1}\bra{j-1}\bra{k-1}\bra{l-1}\ket{i-1}\ket{j}\ket{k-1}\ket{l'-1}\\
    +\sum_{i=1}^{N_x}\sum_{j=1}^{N_y-1}\sum_{k=1}^{N_\xi}\sum_{l=1}^{\frac{N_\eta}{2}}\frac{c_{ij}\eta_{l}}{\Delta y\sqrt{\xi_{k}^{2}+\eta_{l}^{2}}}a^{R}_{i,j+\frac{1}{2},kl}\bra{i-1}\bra{j-1}\bra{k-1}\bra{l-1}\ket{i-1}\ket{j-1}\ket{k-1}\ket{N_\eta-l}\\
    -\sum_{i=1}^{N_x}\sum_{j=1}^{N_y}\sum_{k=1}^{N_\xi}\sum_{l=1}^{\frac{N_\eta}{2}}\frac{c_{ij}\eta_{l}}{\Delta y\sqrt{\xi_{k}^{2}+\eta_{l}^{2}}}\bra{i-1}\bra{j-1}\bra{k-1}\bra{l-1}\ket{i-1}\ket{j-1}\ket{k-1}\ket{l-1},
\end{multline}
and a $N_xN_yN_\xi N_\eta$-dimensional vector 
\begin{multline}\label{equ:inhomogeneoustermx2}
    b_2=-\sum_{i=1}^{N_x}\sum_{k=1}^{N_\xi}\sum_{l=\frac{N_\eta}{2}+1}^{N_\eta}\frac{c_{i0}\eta_{l}}{\Delta y\sqrt{\xi_{k}^{2}+\eta_{l}^{2}}}f_{i0kl}(t)\ket{i-1}\ket{0}\ket{k-1}\ket{l-1}\\
    +\sum_{i=1}^{N_x}\sum_{k=1}^{N_\xi}\sum_{l=1}^{\frac{N_\eta}{2}}\frac{c_{i,N_y}\eta_{l}}{\Delta y\sqrt{\xi_{k}^{2}+\eta_{l}^{2}}}f_{i,N_y+1,k,l}(t)\ket{i-1}\ket{N_y-1}\ket{k-1}\ket{l-1}.
\end{multline}
The sparsity in row is easy to compute $s_r(A_2)\le 4$. Denote the set of interfaces vertical to the $y$-axis as $\mathcal{I}_y=\{(x_i,y_{j+\frac{1}{2}}):c(x,y)\text{ is discontinuous on }(x_i,y_{j+\frac{1}{2}})\}$. Then $s_c(A_2)\le2\max_{(x,y)\in\mathcal{I}_y}\left\{\left\lceil\frac{c^+(x,y)}{c^-(x,y)}\right\rceil,\left\lceil\frac{c^-(x,y)}{c^+(x,y)}\right\rceil\right\}+2:=Q_2$ .


For the term $-\frac{c_{i+\frac{1}{2},j}^{-}-c_{i-\frac{1}{2},j}^{+}}{\Delta x\Delta\xi}\sqrt{\xi_{k}^{2}+\eta_{l}^{2}}\left(f_{ij,k+\frac{1}{2},l}-f_{ij,k-\frac{1}{2},l}\right)$ in \eqref{equ:2dimLiouvilleOpticsSmeidis}, define 
\begin{equation*}
    d_{ijkl}^x=-\frac{c_{i+\frac{1}{2},j}^{-}-c_{i-\frac{1}{2},j}^{+}}{\Delta x\Delta\xi}\sqrt{\xi_{k}^{2}+\eta_{l}^{2}}.
\end{equation*}
Since the numerical fluxes $f_{ij,k+\frac{1}{2},l},f_{ij,k-\frac{1}{2},l}$ is defined using the upwind discretization, like the discussion in section \ref{sec:lin_advec},
\begin{equation*}
    d_{ijkl}^x\left(f_{ij,k+\frac{1}{2},l}-f_{ij,k-\frac{1}{2},l}\right)=-\frac{|d_{ijkl}^x|+d_{ijkl}^x}{2}f_{ij,k-1,l}(t)+|d_{ijkl}^x|f_{ijkl}(t)-\frac{|d_{ijkl}^x|-d_{ijkl}^x}{2}f_{ij,k+1,l}(t).
\end{equation*}
When $d_{ijkl}^x>0$ (resp. $d_{ijkl}^x<0$) for any $k$ and $i,j,l$ fixed, we impose inflow boundary condition at $(x_i,y_j,\xi_0,\eta_l)$ (resp. $(x_i,y_j,\xi_{N_\xi+1},\eta_l)$) and outflow boundary condition at $(x_i,y_j,\xi_{N_\xi+1},\eta_l)$ (resp. $(x_i,y_j,\xi_0,\eta_l)$).
Write the above equations in the form $A_3\boldsymbol{f}+b_3$ of a product of a matrix and a vector with an inhomogeneous term representing boundary condition, where the equations and variables are in order of $\boldsymbol{f}$, we can get a $N_xN_yN_\xi N_\eta\times N_xN_yN_\xi N_\eta$ matrix 
\begin{multline}\label{equ:mtx3}
    A_3=-\sum_{i=1}^{N_x}\sum_{j=1}^{N_y}\sum_{k=2}^{N_\xi}\sum_{l=1}^{N_\eta}\frac{|d_{ijkl}^x|+d_{ijkl}^x}{2}\bra{i-1}\bra{j-1}\bra{k-1}\bra{l-1}\ket{i-1}\ket{j-1}\ket{k-2}\ket{l-1}\\
    +\sum_{i=1}^{N_x}\sum_{j=1}^{N_y}\sum_{k=1}^{N_\xi}\sum_{l=1}^{N_\eta}|d_{ijkl}^x|\bra{i-1}\bra{j-1}\bra{k-1}\bra{l-1}\ket{i-1}\ket{j-1}\ket{k-1}\ket{l-1}\\ 
    -\sum_{i=1}^{N_x}\sum_{j=1}^{N_y}\sum_{k=1}^{N_\xi-1}\sum_{l=1}^{N_\eta}\frac{|d_{ijkl}^x|-d_{ijkl}^x}{2}\bra{i-1}\bra{j-1}\bra{k-1}\bra{l-1}\ket{i-1}\ket{j-1}\ket{k}\ket{l-1},
\end{multline}
and a $N_xN_yN_\xi N_\eta$-dimensional vector
\begin{multline}\label{equ:inhomogeneoustermxi1}
    b_3=-\sum_{i=1}^{N_x}\sum_{j=1}^{N_y}\sum_{l=1}^{N_\eta}\frac{|d_{ij0l}^x|+d_{ij0l}^x}{2}f_{ij0l}(t)\ket{i-1}\ket{j-1}\ket{0}\ket{l-1}\\ 
    -\sum_{i=1}^{N_x}\sum_{j=1}^{N_y}\sum_{l=1}^{N_\eta}\frac{|d_{ij,N_\xi+1,l}^x|-d_{ij,N_\xi+1,l}^x}{2}f_{ij,N_\xi+1,l}(t)\ket{i-1}\ket{j-1}\ket{N_\xi-1}\ket{l-1}.
\end{multline}

Similarly, for the term $-\frac{c_{i,j+\frac{1}{2}}^{-}-c_{i,j-\frac{1}{2}}^{+}}{\Delta y\Delta\eta}\sqrt{\xi_{k}^{2}+\eta_{l}^{2}}\left(f_{ijk,l+\frac{1}{2}}-f_{ijk,l-\frac{1}{2}}\right)$ in \eqref{equ:2dimLiouvilleOpticsSmeidis}, define 
\begin{equation*}
    d_{ijkl}^y=-\frac{c_{i,j+\frac{1}{2}}^{-}-c_{i,j-\frac{1}{2}}^{+}}{\Delta y\Delta\xi}\sqrt{\xi_{k}^{2}+\eta_{l}^{2}}.
\end{equation*}
Then
\begin{equation*}
    d_{ijkl}^y\left(f_{ijk,l+\frac{1}{2}}-f_{ijk,l-\frac{1}{2}}\right)=-\frac{|d_{ijkl}^y|+d_{ijkl}^y}{2}f_{ijk,l-1}(t)+|d_{ijkl}^y|f_{ijkl}(t)-\frac{|d_{ijkl}^y|-d_{ijkl}^y}{2}f_{ijk,l+1}(t).
\end{equation*}
Write the above equations in the form $A_4\boldsymbol{f}+b_4$ of a product of a matrix and a vector with an inhomogeneous term representing boundary condition, where the equations and variables are in order of $\boldsymbol{f}$, we can get a $N_xN_yN_\xi N_\eta\times N_xN_yN_\xi N_\eta$ matrix 
\begin{multline}\label{equ:mtx4}
    A_4=-\sum_{i=1}^{N_x}\sum_{j=1}^{N_y}\sum_{k=1}^{N_\xi}\sum_{l=2}^{N_\eta}\frac{|d_{ijkl}^y|+d_{ijkl}^y}{2}\bra{i-1}\bra{j-1}\bra{k-1}\bra{l-1}\ket{i-1}\ket{j-1}\ket{k-1}\ket{l-2}\\
    +\sum_{i=1}^{N_x}\sum_{j=1}^{N_y}\sum_{k=1}^{N_\xi}\sum_{l=1}^{N_\eta}|d_{ijkl}^y|\bra{i-1}\bra{j-1}\bra{k-1}\bra{l-1}\ket{i-1}\ket{j-1}\ket{k-1}\ket{l-1}\\ 
    -\sum_{i=1}^{N_x}\sum_{j=1}^{N_y}\sum_{k=1}^{N_\xi}\sum_{l=1}^{N_\eta-1}\frac{|d_{ijkl}^y|-d_{ijkl}^y}{2}\bra{i-1}\bra{j-1}\bra{k-1}\bra{l-1}\ket{i-1}\ket{j-1}\ket{k-1}\ket{l},
\end{multline}
and a $N_xN_yN_\xi N_\eta$-dimensional vector
\begin{multline}\label{equ:inhomogeneoustermxi2}
    b_4=-\sum_{i=1}^{N_x}\sum_{j=1}^{N_y}\sum_{k=1}^{N_\xi}\frac{|d_{ijk0}^y|+d_{ijk0}^y}{2}f_{ijk0}(t)\ket{i-1}\ket{j-1}\ket{k-1}\ket{0}\\ 
    -\sum_{i=1}^{N_x}\sum_{j=1}^{N_y}\sum_{k=1}^{N_\xi}\frac{|d_{ijk,N_\eta+1}^y|-d_{ijk,N_\eta+1}^y}{2}f_{ijk,N_\eta+1}(t)\ket{i-1}\ket{j-1}\ket{k-1}\ket{N_\eta-1}.
\end{multline}

Combining \eqref{equ:mtx1}, 
\eqref{equ:mtx2}, 
\eqref{equ:mtx3}, \eqref{equ:mtx4}, \eqref{equ:inhomogeneoustermx1}, \eqref{equ:inhomogeneoustermx2}, \eqref{equ:inhomogeneoustermxi1} and \eqref{equ:inhomogeneoustermxi2}, we can get the ODE $\boldsymbol{f}'(t)=A\boldsymbol{f}(t)+\boldsymbol{b}(t)$ where $A=-A_1-A_2-A_3-A_4$ and $\boldsymbol{b}(t)=-b_1-b_2-b_3-b_4$.

\subsection{Schr\"odingerization}\label{sec:interface1Schro2d}
We first write the above ODE into the same homogeneous form \eqref{equ:interface1homogeneous} where $\boldsymbol{0}$ is a $N_xN_yN_\xi N_\eta\times N_xN_yN_\xi N_\eta$ matrix all of $0$ and $\boldsymbol{1}$ is a $N_xN_yN_\xi N_\eta$-dimensional vector all of $1$. The later process are the same as \eqref{equ:procedure}.

The following theorem give the complexity of the Schr\"odingerizaiton based Hamiltonian-preserving scheme  in two dimensional case.

\begin{theorem}\label{thm:2dim_complexity1}
    In the 2-dimensional case, given the initial quantum state $\ket{\boldsymbol{u}(0)}$, assume that the $x$ grid number, $y$ grid number, $\xi$ grid number, $\eta$ grid number and extended $p$ grid number are $N_x=2^{n_x}$, $N_y=2^{n_y}$, $N_\xi=2^{n_\xi}$, $N_\eta=2^{n_\eta}$ and $N_p=2^{n_p}$ respectively, the inhomogeneous term $\boldsymbol{b}$ is independent of time and the error of Hamiltonian-preserving scheme, spectral method in $p$ and Hamiltonian simulation are all $\epsilon$. With the Schr\"odingerization method, the state $\ket{\boldsymbol{u}(t)}$ can be simulated to time $T$, and success probability is at least $1-2\epsilon$ with
    \begin{itemize}
        \item $\mathcal{O}(\max\{Q_1,Q_2\}T\epsilon^{-2}+\frac{\log\epsilon^{-1}}{\log\log\epsilon^{-1}})$ queries and a factor $\mathcal{O}(5\log(\epsilon^{-1})+\log(\epsilon^{-1})\polylog(\log(\epsilon^{-1})))$ additional quantum gates, when solution has discontinuities only at interfaces;
        \item $\mathcal{O}(\max\{Q_1,Q_2\}T\epsilon^{-3}+\frac{\log\epsilon^{-1}}{\log\log\epsilon^{-1}})$ queries and a factor $\mathcal{O}(9\log(\epsilon^{-1})+\log(\epsilon^{-1})\polylog(\log(\epsilon^{-1})))$ additional quantum gates, when solution has discontinuities not only at interfaces,
    \end{itemize}
    where 
    \[
Q_1:=2\max_{(x,y)\in\mathcal{I}_x}\left\{\left\lceil\frac{c^+(x,y)}{c^-(x,y)}\right\rceil,\left\lceil\frac{c^-(x,y)}{c^+(x,y)}\right\rceil\right\}+2\ge s_c(A_1)\ge 2.
\]
and 
\[
Q_2:=2\max_{(x,y)\in\mathcal{I}_y}\left\{\left\lceil\frac{c^+(x,y)}{c^-(x,y)}\right\rceil,\left\lceil\frac{c^-(x,y)}{c^+(x,y)}\right\rceil\right\}+2\ge s_c(A_2)\ge 2.
\]
\end{theorem}
\begin{proof}
    We use Lemma \ref{lem:complexity} to give the proof. Matrix elements are specified to $\log(\epsilon^{-1})$ bits of precision.
    
    From previous discussion, 
    \begin{equation*}
        \begin{aligned}
            s(H)=s(H_1\otimes D_\mu-H_2\otimes I)&\le \max\{s(A+A^\dagger),s(A-A^\dagger)\}+s(\text{diag}(\boldsymbol{b}(t))/\epsilon)\\
            &\le \max\{s(A_1+A_1^\dagger),s(A_1-A_1^\dagger)\}+\max\{s(A_2+A_2^\dagger),s(A_2-A_2^\dagger)\}\\
            &\quad +\max\{s(A_3+A_3^\dagger),s(A_3-A_3^\dagger)\}+\max\{s(A_4+A_4^\dagger),s(A_4-A_4^\dagger)\}\\ 
            &\quad -3+1\\
            &\le (s_r(A_1)+s_c(A_1)-1)+(s_r(A_2)+s_c(A_2)-1)+(3)\\
            &\le \max\{Q_1,Q_2\}+3+3+3=\max\{Q_1,Q_2\}+9.
        \end{aligned}
    \end{equation*}
    In the $p$-direction, the initial value is only continuous but not differentiable. According to spectral method, in order to reach precision $\epsilon$, the total number of $p$ grid points is $N_p=\mathcal{O}(\epsilon^{-1})$.
    The total number of points in $x$ and $\xi$-direction is related to where the discontinuities lies.
    \begin{itemize}
        \item When solution has discontinuities only at interfaces, the Hamiltonian-preserving scheme can give the first order precision. Thus $N_x=N_y=N_\xi=N_\eta=\mathcal{O}(\epsilon^{-1})$ and $\|H\|_{\max}\le\|H_{1}\otimes D_{p}\|_{\max}+\|H_{2}\|_{\max}\le\mathcal{O}(\left\|H_{1}\right\|_{\max}N_p)=\mathcal{O}((N_x+N_y+N_\xi+N_\eta)N_p)=\mathcal{O}(\epsilon^{-2})$. The query complexity is 
        \begin{equation*}
            \mathcal{O}(\max\{Q_1,Q_2\}T\epsilon^{-2}+\frac{\log\epsilon^{-1}}{\log\log\epsilon^{-1}}).
        \end{equation*} 
        Since the quantum Fourier transformation in Schr\"odingerization can be implemented by using $\mathcal{O}(m\log m)$ gates \cite{namApproximateQuantumFourier2020}, the number of overall additional quantum gates is 
        \begin{equation*}
            \mathcal{O}(5\log(\epsilon^{-1})+\log(\epsilon^{-1})\polylog(\log(\epsilon^{-1}))+\log(\epsilon^{-1})\log(\log(\epsilon^{-1}))).
        \end{equation*}
        \item When solution has discontinuities not only at interfaces but also elsewhere, the Hamiltonian-preserving scheme can give the halfth order precision. Thus $N_x=N_y=N_\xi=N_\eta=\mathcal{O}(\epsilon^{-2})$ and $\|H\|_{\max}\le\mathcal{O}(\epsilon^{-3})$. The query complexity is 
        \begin{equation*} 
            \mathcal{O}(\max\{Q_1,Q_2\}T\epsilon^{-3}+\frac{\log\epsilon^{-1}}{\log\log\epsilon^{-1}}).
        \end{equation*} 
        and the number of overall additional quantum gates is 
        \begin{equation*}
            \mathcal{O}(9\log(\epsilon^{-1})+\log(\epsilon^{-1})\polylog(\log(\epsilon^{-1}))+\log(\epsilon^{-1})\log(\log(\epsilon^{-1}))).
        \end{equation*}
    \end{itemize}
\end{proof}

\begin{remark}
    In the classical implementation of the Hamiltonian-preserving schemes, one can use any time discretization for the time derivative. Here we use the forward Euler time descretization as an example to analyze the classical complexity which contains the number of products and quotients. After discretization, the ODE $\boldsymbol{f}'(t)=A\boldsymbol{f}(t)+\boldsymbol{b}(t)$ becomes $\boldsymbol{f}^{n+1}=(\Delta tA+I)\boldsymbol{f}^n+\Delta t\boldsymbol{b}$ for $n=0,\ldots,\frac{T}{N_t}-1$. To satisfy the hyperbolic CFL condition
    \begin{equation*}
         \Delta t\max_{i,j,k,l}\left[\frac{c_{ij}\xi_{k}}{\Delta x\sqrt{\xi_{k}^{2}+\eta_{l}^{2}}}+\frac{c_{ij}\eta_{l}}{\Delta y\sqrt{\xi_{k}^{2}+\eta_{l}^{2}}}+\frac{\left|c_{i+\frac{1}{2},j}^{-}-c_{i-\frac{1}{2},j}^{+}\right|}{\Delta x\Delta\xi}\sqrt{\xi_{k}^{2}+\eta_{l}^{2}}+\frac{\left|c_{i,j+\frac{1}{2}}^{-}-c_{i,j-\frac{1}{2}}^{+}\right|}{\Delta y\Delta\eta}\sqrt{\xi_{k}^{2}+\eta_{l}^{2}}\right]\leqslant1,
    \end{equation*}
    a time step $\Delta t=\mathcal{O}(\Delta x,\Delta y,\Delta \xi,\Delta \eta)$ is needed \cite{jinHamiltonianpreservingSchemesLiouville2006}, so we take $N_t=\mathcal{O}(TN_x)$. Then the classical complexity is 
    \begin{equation}
        \mathcal{O}(s_r(\Delta tA+I)N_xN_y N_\xi N_\eta N_t)=\begin{cases}
            \mathcal{O}(7T\epsilon^{-5}) & \text{when solution has discontinuities only at interfaces},\\
            \mathcal{O}(7T\epsilon^{-10})& \text{when solution has discontinuities not only at interfaces}.
        \end{cases}
    \end{equation}
    There is a polynomial advantage in $\epsilon$ for quantum algorithms.
\end{remark}

\section{Numerical examples}\label{sec:Numerical}
In this section, we present numerical examples to demonstrate the validity of the Schr\"odingerization based Hamiltonian-preserving schemes. All numerical tests are carried out on a classical computer, so the Hamiltonian simulation $\e^{-iHt}$ is replaced with backward Euler or Crank-Nicolson (CN) method. Due to memory limitations (maximum 512 GB), although we selected the maximum number of grid points within the available memory, the numerical resulution is not high. Consequently, the numerical results appear relatively smooth across discontnuities. Nevertheless, the goal of these experiments is to demonstrate the correctness of the method,  therefore, what we will present below suffice for this purpose. 

\begin{example}(\cite{jinHamiltonianPreservingSchemeLiouville2006})\label{ex:I2ex1}
    A 1D problem with bounded solution. Consider the 1D Liouville equation \eqref{equ:1dimLiouvilleOptics}
    with a discontinuous wave speed given by
    \begin{equation}
    c(x)=\begin{cases}
        0.6,&x<0,\\
        0.2,&x>0.
    \end{cases}
    \end{equation}
    The initial data are given by
    \begin{equation}
    f(x,\xi,0)=\begin{cases}
        1,&x<0,\xi>0,\sqrt{x^2+4\xi^2}<1,\\
        1,&x>0,\xi<0,\sqrt{x^2+\xi^2}<1,\\
        0,&\mathrm{otherwise}.
    \end{cases}
    \end{equation}
\end{example}
In this example the reflection and transmission coefficients $\alpha^R,\alpha^T$ at the interface are $\alpha^{R}=\frac{1}{4},\alpha^{T}=\frac{3}{4}.$ The exact solution for $f$ at $t=1$ is given by
\begin{equation}
    f(x,\xi,1)=\begin{cases}
        \alpha^T,&0<x<0.2,\quad\sqrt{1-(0.2-x)^2}<\xi<1.5\sqrt{1-(3x-0.6)^2};\\
        1,&0<x<0.2,\quad0<\xi<\sqrt{1-(0.2-x)^2};\\
        1,&0<x<0.8,\quad-\sqrt{1-(x+0.2)^2}<\xi<0;\\
        1,&-0.4<x<0,\quad0<\xi<\frac{1}{2}\sqrt{1-(x-0.6)^2};\\
        1,&-0.6<x<0,\quad-\frac{1}{3}\sqrt{1-\left(\frac{x}{3}+0.2\right)^2}<\xi<0;\\
        \alpha^R,&-0.6<x<0,\quad-\frac{1}{2}\sqrt{1-(x+0.6)^2}<\xi<-\frac{1}{3}\sqrt{1-\left(\frac{x}{3}+0.2\right)^2};\\
        0,&\text{otherwise,}
    \end{cases}
\end{equation}
as shown in the first column of Figure \ref{fig:I2ex1solution}.

We are also interested in computing the moments of $f$, which include the density
\begin{equation}\label{equ:density}
    \rho(x,t)=\int f(x,\xi,t)\d\xi
\end{equation}
and the averaged slowness
\begin{equation}\label{equ:averagedslowness}
    u(x,t)=\frac{\int f(x,\xi,t)\xi \d\xi}{\rho(x,t)}.
\end{equation}

At $t = 1$, the exact density is
\begin{equation}
    \rho(x,1)=\begin{cases}
    \sqrt{1-(x+0.2)^2}&0.2<x<0.8;\\
    1.5\alpha^T\sqrt{1-(3x-0.6)^2}+\alpha^R\sqrt{1-(0.2-x)^2}+\sqrt{1-(x+0.2)^2}&0<x<0.2;\\
    \frac{\alpha^T}{3}\sqrt{1-\left(\frac{x}{3}+0.2\right)^2}+\frac{\alpha^R}{2}\sqrt{1-(x+0.6)^2}&-0.6<x<-0.4;\\
    \frac{\alpha^T}{3}\sqrt{1-\left(\frac{x}{3}+0.2\right)^2}+\frac{\alpha^R}{2}\sqrt{1-(x+0.6)^2}+\frac{1}{2}\sqrt{1-(x-0.6)^2}&-0.4<x<0,\\
    0&\text{otherwise.}\end{cases}
\end{equation}
The averaged slowness  has support within  $[-0.6, 0.8]$, since the density is zero outside this interval. The exact averaged slowness in $[-0.6, 0.8]$ is
\begin{equation}
    u(x,1)=\frac{1}{2\rho(x,1)}\begin{cases}
        -\left[1-(x+0.2)^2\right]&0.2<x<0.8;\\
        2.25\alpha^T\left[1-(3x-0.6)^2\right]+\alpha^R\left[1-(0.2-x)^2\right]-\left[1-(x+0.2)^2\right]&0<x<0.2;\\
        \frac{-\alpha^T}{9}\left[1-\left(\frac{x}{3}+0.2\right)^2\right]-\frac{\alpha^R}{4}\left[1-(x+0.6)^2\right]&-0.6<x<-0.4;\\
        \frac{-\alpha^T}{9}\left[1-\left(\frac{x}{3}+0.2\right)^2\right]-\frac{\alpha^R}{4}\left[1-(x+0.6)^2\right]+\frac{1}{4}\left[1-(x-0.6)^2\right]&-0.4<x<0.
    \end{cases}
\end{equation}

The computational domain is chosen as $[x, \xi] \in [-1.5, 1.5]\times[-1.6, 1.6]$. The cell number is set by $N_x=N_\xi=2^7$. In the extended space $p$, since $\lambda_{\mathrm{max}}^{+}(H_1)=0.7434$ and $\lambda_{\mathrm{max}}^{-}(H_1)=50.7709$, we take the truncated interval in $p$ as $[-55.771,5.7454]$ with $N_p=2^{14}$ and $\Delta p=0.00375$. We choose the recovery point $p=0.7454\ge p^\diamond = \lambda^+_{\max}(H_1)T=0.7434$ and time step as $\Delta t =0.02$. 

Figure \ref{fig:I2ex1solution} shows the exact, classical and quantum numerical solutions  $f$. Figure \ref{fig:I2ex1rhou} shows the classical and quantum numerical density $\rho$ and averaged slowness $u$ along with the exact solutions in the physical space.

\begin{figure}[!htb]
    \centering
    \includegraphics[width=1.\linewidth]{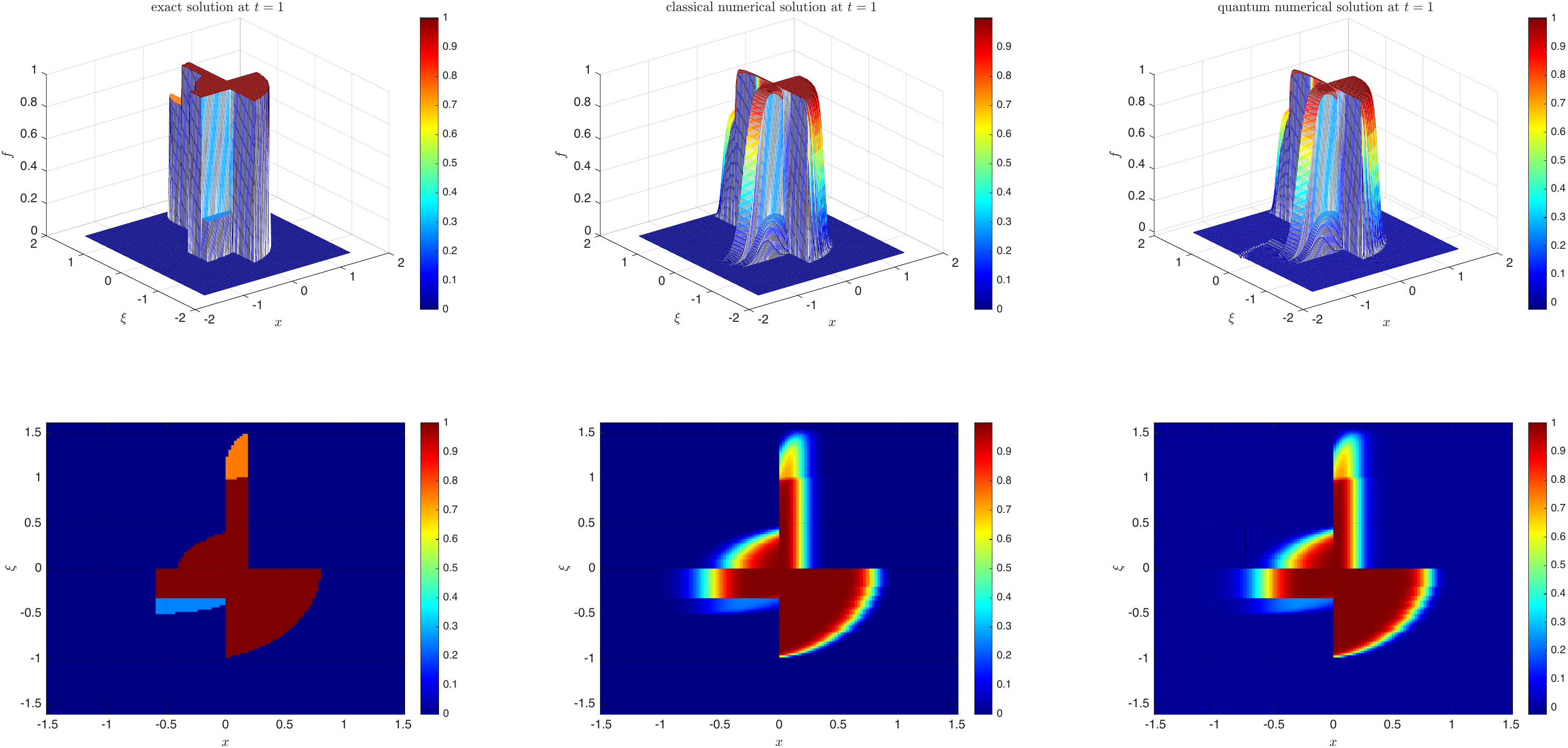}
    \caption{Example \ref{ex:I2ex1}, the density distribution function $f(x,\xi,t)$ at $t=1$. First row: 3D plot; second row: contour plot. First column: the exact solution; second column: the classical numerical solution; third column: the Schr\"odingerization solution.}
    \label{fig:I2ex1solution}
\end{figure}

\begin{figure}[!htb]
    \centering
    \includegraphics[width=1.\linewidth]{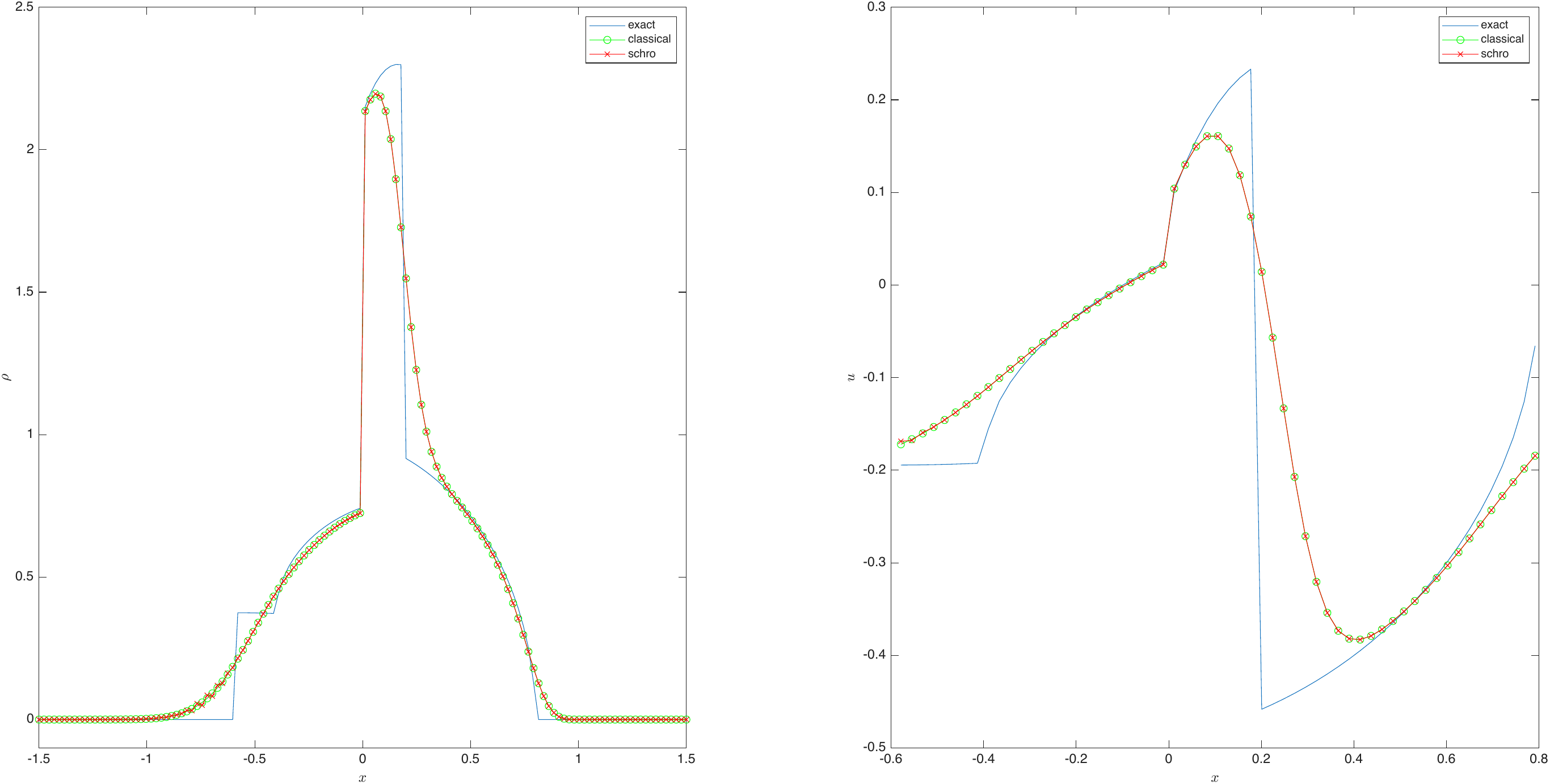}
    \caption{Example \ref{ex:I2ex1}, the density $\rho$ and averaged slowness $u$ at $t = 1$. Solid blue line: the exact solution; green ``o": the classical numerical solution; red solid line with ``$\times$": the Schr\"odingerization solution. Left: the density $\rho$; Right: the averaged slowness $u$.}
    \label{fig:I2ex1rhou}
\end{figure}

\begin{example}(\cite{jinHamiltonianPreservingSchemeLiouville2006})\label{ex:I2ex2}
    Computing the physical observables of a 1D problem with measure-valued solution. Consider the 1D Liouville equation \eqref{equ:1dimLiouvilleOptics}, where the wave speed is a well-shaped function
    \begin{equation}
    c(x)=\begin{cases}0.6,&\quad-0.4<x<0.4\\1,&\quad\mathrm{else}\end{cases}
    \end{equation}
    and the initial data is a delta-function
    \begin{equation}
    f(x,\xi,0)=\delta(\xi-w(x))
    \end{equation}
    with
    \begin{equation}
    w(x)=\begin{cases}0.5,&{x\leq-1.6;}\\{0.5-\frac{0.4}{(1.6)^{2}}(x+1.6)^{2},}&{-1.6<x\leq0;}\\{-0.5+\frac{0.4}{(1.6)^{2}}(x-1.6)^{2},}&{0<x<1.6;}\\{-0.5,}&{x\geq1.6,}\end{cases}
    \end{equation}
    which is ploted as blue dashed lines in the first column of Figure \ref{fig:I2ex2solution}.
\end{example}
In this example we are interested in the approximation of the density \eqref{equ:density} and the averaged slowness \eqref{equ:averagedslowness}.

In the computation, we first approximate the delta function initial data $f(x,\xi,0)$ by the simple piecewise linear kernel \cite{engquistDiscretizationDiracDelta2005}:
\begin{equation}\label{equ:discretedelta}
    \delta_\beta^{(1)}(x)=\begin{cases}\frac{1}{\beta}\left(1-\left|\frac{x}{\beta}\right|\right),&|\frac{x}{\beta}|\leq1,\\
        0,&|\frac{x}{\beta}|>1.\end{cases}
\end{equation}
If $\left|\xi_j-w(x_i)\right|<\beta$, set $f_{ij}^0=\frac1\beta\left(1-|\frac{\xi_j-w(x_i)}\beta|\right)$, and $f_{ij}^0=0$ otherwise. 
Direct numerical methods (DNM) for the Liouville equation with measure-valued initial data , which approximate the initial delta function first, then evolve the Liouville equation, could suffer from a poor numerical resolution due to the numerical approximation of the initial data of delta function as well as numerical dissipation.
Numerical viscosity would smear out the delta-function and reduce the accuracy when approximating  the integral, which is  needed for total density computation.
The level set method proposed in \cite{jinComputingMultivaluedPhysical2005,jinComputingMultivaluedPhysical2005a} decomposes the density distribution $f$ into the bounded level set functions obeying the same Liouville equation, which greatly enhances the numerical resolution. One only involves numerically the delta function at the output time when the moments—which has delta functions in their integrands—need to be evaluated numerically. 
However, the extension of this density distribution decomposing approach to the case of partial transmission and reflection is not straightforward. In particular, as the number of transmissions and reflections increase in time, so does the number of needed level set functions satisfying \eqref{equ:ddimLiouvilleOptics}. Example \ref{ex:I1ex2} will use the decomposition method to the case of full transmission.

We then use the Hamiltonian-preserving scheme to solve the Liouville equation \eqref{equ:1dimLiouvilleOptics}. Then the moments are recovered by
\begin{equation}
    \rho_i^n=\sum_jf_{ij}^n\Delta\xi,\quad u_i^n=\left(\sum_jf_{ij}^n\xi_j\Delta\xi\right)/\rho_i^n.
\end{equation}

With partial transmissions and reflections, the exact multivalued slowness at $t = 1$ 
is depicted by the red solid line in the first column of Figure \ref{fig:I2ex2solution}.

In this example the reflection and transmission coefficients $\alpha^R, \alpha^T$ at the wave speed interface are $\alpha^R = \frac{1}{16} , \alpha^T = \frac{15}{16}$ . At $t = 1$, the exact density and averaged
\begin{equation}
    \rho(x,1)=\begin{cases}
    1,&-1.6<x<-1.4;\\
    1+\alpha^R,&-1.4<x<-0.4-1/3;\\
    1+\alpha^R+0.6\alpha^T,&-0.4-1/3<x<-0.4;\\
    1+\alpha^R+\alpha^T/0.6,&-0.4<x<-0.2;\\
    \alpha^T/0.3,&-0.2<x<0.2;\\
    1+\alpha^R+\alpha^T/0.6,&0.2<x<0.4;\\
    1+\alpha^R+0.6\alpha^T,&0.4<x<0.4+1/3;\\
    1+\alpha^R,&0.4+1/3<x<1.4;\\
    1,&1.4<x<1.6;
\end{cases}
\end{equation}
and
\begin{equation}
    u(x,1)=\frac{1}{\rho(x,1)}\begin{cases}
        0.5,&-1.6<x<-1.4;\\
        0.5-\alpha^R\Upsilon(x+0.2),&-1.4<x<-0.4-\frac{1}{3};\\
        0.5-\alpha^R\Upsilon(x+0.2)-0.36\alpha^T\Upsilon(0.6x-1.16),&-0.4-\frac{1}{3}<x<-0.6;\\
        \Upsilon(x+0.6)-\alpha^R\Upsilon(x+0.2)-0.36\alpha^T\Upsilon(0.6x-1.16),&-0.6<x<-0.4;\\
        \frac{\alpha^T}{0.36}\Upsilon(\frac{x}{6}+\frac{13}{15})-\Upsilon(x-1)+\alpha^R\Upsilon(x+1.8),&-0.4<x<-0.2;\\
        \frac{\alpha^T}{0.36}\Upsilon(\frac{x}{6}+\frac{13}{15})-\frac{\alpha^T}{0.36}\Upsilon(\frac{x}{0.6}-\frac{13}{15})&-0.2<x<0.2;\\
        -\alpha^T\Upsilon(\frac{x}{0.6}-\frac{13}{15})+\Upsilon(x+1)-\alpha^R\Upsilon(x-1.8),&0.2<x<0.4;\\
        -\Upsilon(x-0.6)+\alpha^R\Upsilon(x-0.2)+0.36\alpha^T\Upsilon(0.6x+1.16),&0.4<x<0.6;\\
        -0.5+\alpha^R\Upsilon(x-0.2)+0.36\alpha^T\Upsilon(0.6x+1.16),&0.6<x<0.4+\frac{1}{3};\\
        -0.5+\alpha^R\Upsilon(x-0.2),&0.4+\frac{1}{3}<x<1.4.\\
        -0.5,&0.4<x<1.6;
    \end{cases}
\end{equation}
with $\Upsilon(x)=0.5-\frac{0.4}{(1.6)^2}x^2$.

The computational domain is chosen as $[x, \xi] \in [-1.5, 1.5]\times[-1, 1]$. The cell number is set by $N_x=N_\xi=2^7$. In the extended spcace $p$, since $\lambda_{\mathrm{max}}^{+}(H_1)=0.6006$ and $\lambda_{\mathrm{max}}^{-}(H_1)=79.295$, we take the truncated interval in $p$ as $[-84.295,5.6006]$ with $N_p=2^{14}$ and $\Delta p=0.00375$. We choose the recovery point $p=0.6021\ge p^\diamond = \lambda^+_{\max}(H_1)T=0.6006$ and time step as $\Delta t =0.02$. 
We use a linear relation between $\beta$ and the mesh size $\Delta\xi$: $\beta=\Delta\xi$, which guarantees the numerical convergence. 

Figure \ref{fig:I2ex2solution} shows the exact, classical and quantum numerical solutions $f$. Figure \ref{fig:I2ex2rhou} shows the classical and quantum numerical density $\rho$ and averaged slowness $u$ along with the exact solutions in the physical space.

\begin{figure}[!htb]
    \centering
    \includegraphics[width=1.\linewidth]{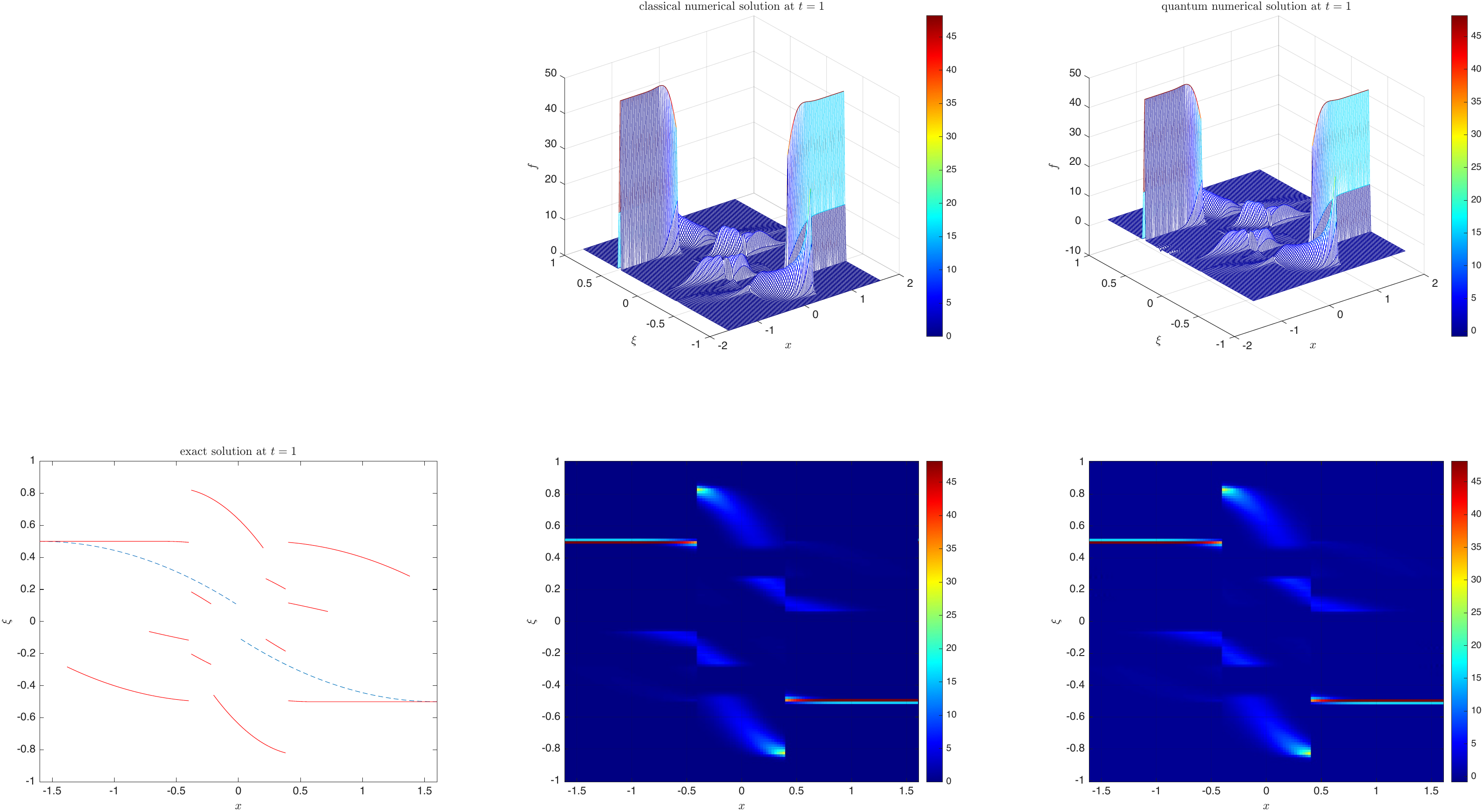}
    \caption{Example \ref{ex:I2ex2}, the density distribution function $f(x,\xi,t)$ at $t=1$. First row: 3D plot; second row: contour plot. First column: the exact solution; second column: the classical numerical solution; third column: the Schr\"odingerization solution.}
    \label{fig:I2ex2solution}
\end{figure}

\begin{figure}[!htb]
    \centering
    \includegraphics[width=1.\linewidth]{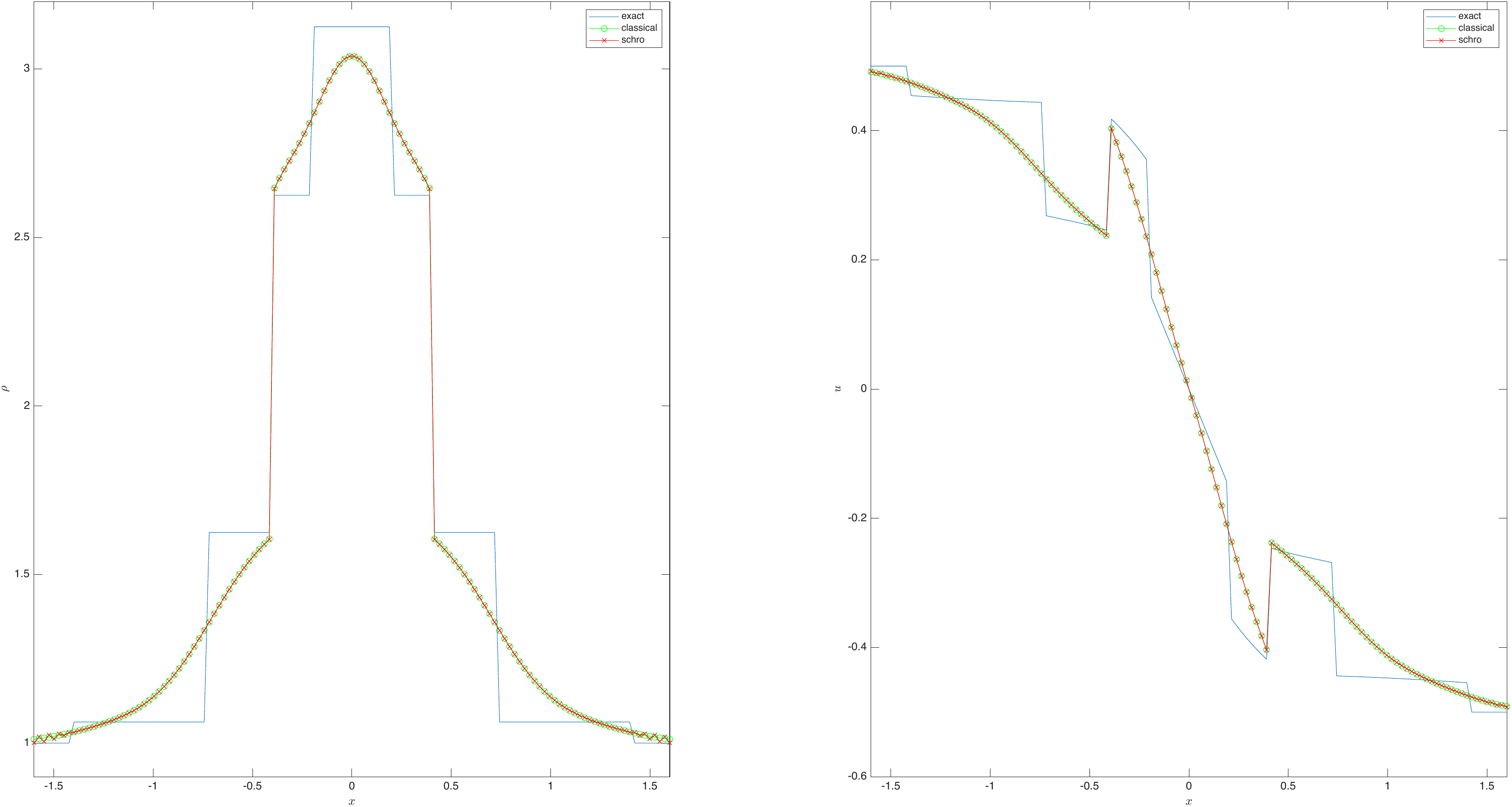}
    \caption{Example \ref{ex:I2ex2}, the density $\rho$ and averaged slowness $u$ at $t = 1$. Solid blue line: the exact solution; green ``o": the classical numerical solution; red solid line with ``$\times$": the Schr\"odingerization solution. Left: the density $\rho$; Right: the averaged slowness $u$.}
    \label{fig:I2ex2rhou}
\end{figure}

\begin{example}(\cite{jinHamiltonianpreservingSchemesLiouville2006})\label{ex:I1ex2}
     Computing the physical observables of an 1D problem with measure-valued solution.  Consider the 1D Liouville equation \eqref{equ:1dimLiouvilleOptics}, where the wave speed is 
     \begin{equation*}
        c(x)=\begin{cases}
            \frac{1}{e-1},&x\leqslant-1,\\
            \frac{1}{e-1}+1+x,&-1<x<0,\\
            \frac{1}{e-1}+0.5-x,&0<x<1,\\
            \frac{1}{e-1}-0.5,&x\geqslant1,
        \end{cases}
    \end{equation*}
    the initial data are given by
    \begin{equation}
    f(x,\xi,0)=\delta(\xi-w(x))
    \end{equation}
    with
    \begin{equation}
    w(x)=\begin{cases}
        0.8,&x\leqslant-1.5,\\
        0.8-\frac{0.8}{\left(1.5\right)^2}\left(x+1.5\right)^2,&-1.5<x\leqslant0,\\
        -0.8+\frac{0.8}{\left(1.5\right)^2}\left(x-1.5\right)^2,&0<x<1.5,\\
        -0.8,&x\geqslant1.5,
    \end{cases}
    \end{equation}
    which is ploted as blue dashed lines in the first column of Figure \ref{fig:I1ex2solution}.
    Here we suppose the wave length is much shorter than the width of the interface while both lengths go to zero. Then we only need to condsider the transmission, i.e., $\alpha^T=1$ and $\alpha^R=0$.
\end{example}

In this example we are interested in computing the moments of $f$, which include the density \eqref{equ:density} and the averaged slowness \eqref{equ:averagedslowness}.
These quantities are computed by decomposition techniques described in \cite{jinHamiltonianpreservingSchemesLiouville2006}. We first solve the level set function $\psi$ and modified density function $\phi$ which satisfy the Liouville equation \eqref{equ:1dimLiouvilleOptics} with initial data $\xi-w(x)$ and $1$, respectively. Then the desired physical observables $\rho$ and $u$ are computed from the numerical singular integrals 
\begin{gather}
    \rho(x,t)=\int f(x,\xi,t)\mathrm{d}\xi=\int\phi(x,\xi,t)\delta(\psi)\mathrm{d}\xi,\\
    u(x,t)=\frac{1}{\rho(x,t)}\int f(x,\xi,t)\xi\mathrm{d}\xi=\int\phi(x,\xi,t)\xi\delta(\psi)\mathrm{d}\xi/\rho(x,t),
\end{gather}
which are computed by discrete delta function \eqref{equ:discretedelta}. Thus one only involves approximating the delta-function at the output time!


The exact slowness and corresponding density at $t = 1$ are given in Appendix A of \cite{jinHamiltonianpreservingSchemesLiouville2006}. The red solid line in the first column of Figure \ref{fig:I1ex2solution} shows the exact multivalued slowness.

The computational domain is chosen as $[x, \xi] \in [-1.5, 1.5]\times[-1, 1]$. The cell number is set by $N_x=N_\xi=2^7$. In the extended spcace $p$, since $\lambda_{\mathrm{max}}^{+}(H_1)=73.796$ and $\lambda_{\mathrm{max}}^{-}(H_1)=252.26$, we take the truncated interval in $p$ as $[-257.26,78.796]$ with $N_p=2^{14}$ and $\Delta p=0.020511$. However, we choose the recovery point $p=14.513< p^\diamond = \lambda^+_{\max}(H_1)T=73.796$ and time step as $\Delta t =0.02$. 
We use a linear relation between $\beta$ and the mesh size $\Delta\xi$: $\beta=6\Delta\xi$, which guarantees the numerical convergence. 

Figure \ref{fig:I1ex2solution} shows the exact, classical and quantum numerical solution $f$. Figure \ref{fig:I1ex2rhou} shows the classical and quantum numerical density $\rho$ and averaged slowness $u$ along with the exact solutions in the physical space, which has higher resolution than Figure \ref{fig:I2ex2rhou}.

\begin{figure}[!htb]
    \centering
    \includegraphics[width=1.\linewidth]{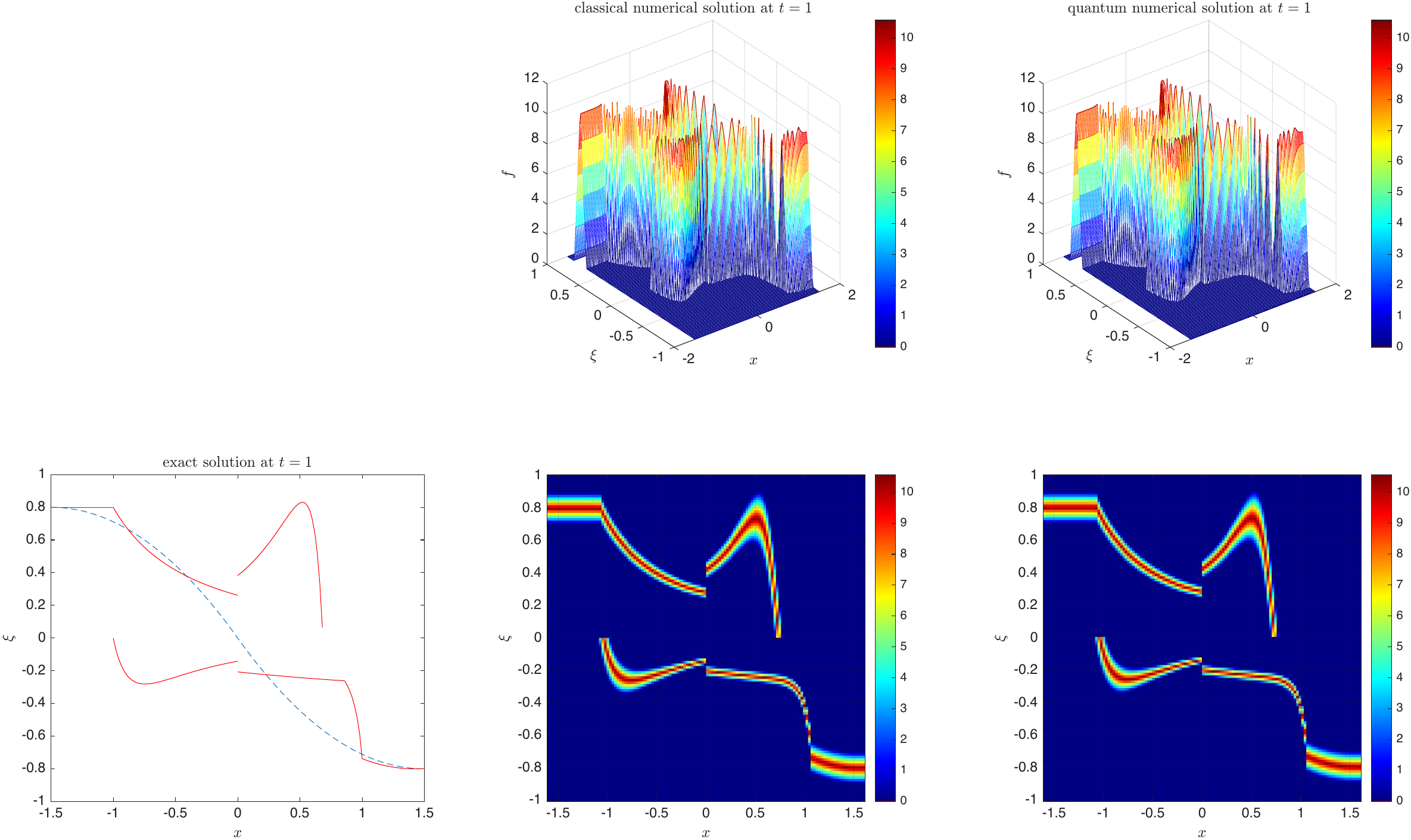}
    \caption{Example \ref{ex:I1ex2}, the density distribution function $f(x,\xi,t)$ at $t=1$. First row: 3D plot; second row: contour plot. First column: the exact solution; second column: the classical numerical solution; third column: the Schr\"odingerization solution.}
    \label{fig:I1ex2solution}
\end{figure}

\begin{figure}[!htb]
    \centering
    \includegraphics[width=1.\linewidth]{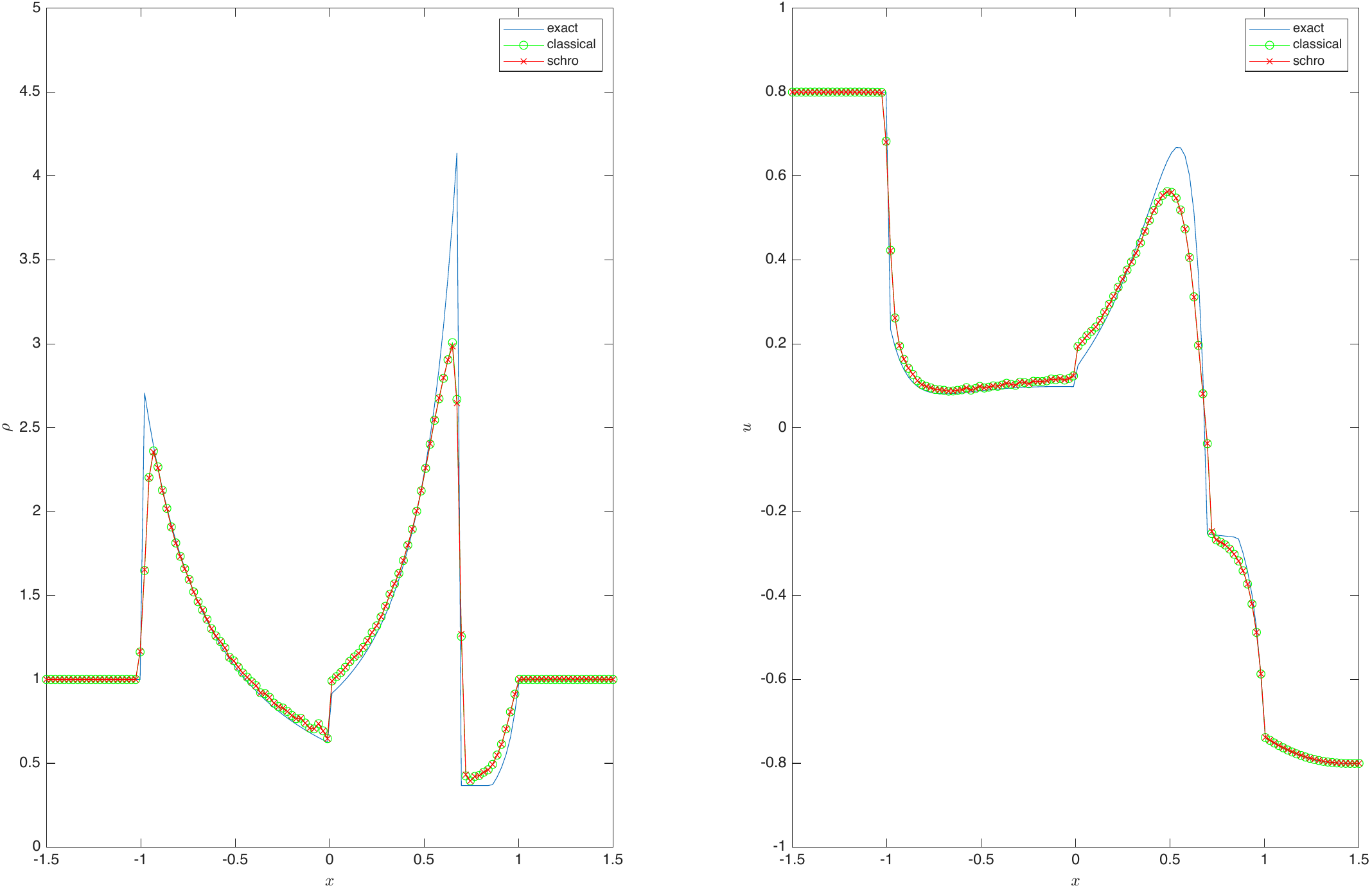}
    \caption{Example \ref{ex:I1ex2}, the density $\rho$ and averaged slowness $u$ at $t = 1$. Solid blue line: the exact solution; green ``o": the classical numerical solution; red solid line with ``$\times$": the Schr\"odingerization solution. Left: the density $\rho$; Right: the averaged slowness $u$.}
    \label{fig:I1ex2rhou}
\end{figure}

\begin{example}\label{ex:I2ex3}
    Computing the physical observables of a 2D problem with a $L^\infty$ solution. Consider the 2D Liouville equation \eqref{equ:2dimLiouvilleOptics} with a discontinuous wave speed
    $$\left.c(x,y)=\left\{\begin{array}{ll}2&\quad y>0\\1&\quad y<0\end{array}\right.\right.$$
    and a smooth initial data
    $$f(x,y,\xi,\eta,0)=\frac1{\pi c_3c_4}\mathrm{exp}\left(-\left(\frac x{c_1}\right)^2-\left(\frac{y+0.1}{c_2}\right)^2-\left(\frac\xi{c_3}\right)^2-\left(\frac{\eta-0.1}{c_4}\right)^2\right),$$
    where $c_1=0.03,c_3=0.05,c_2=c_4=0.025.$
\end{example}

In this example we aim at computing the density which is the zeroth moment of the density distribution
\begin{equation}\label{equ:density2d}
    \rho(x,y,t)=\iint f(x,y,\xi,\eta,t)d\xi d\eta.
\end{equation}

The reflection and transmission coefficients $\alpha^R, \alpha^T$ at the interface are given by \eqref{equ:coefficient}. The ``exact'' solution of $\rho$ is obtained by first solving for $f (x, y, \xi, \eta, t)$ analytically, and then evaluating the integral \eqref{equ:density2d} on a very fine mesh in the $(\xi, \eta)$ space, which is shown in the left column of Figure \ref{fig:I2ex3rhou}.

The computational domain is chosen to be $[x,y,\xi,\eta]\in[-0.12,0.12]\times[-0.2,0.2]\times[-0.2,0.2]\times[-0.2,0.2]$. The cell number is set by $N_x=N_y=N_\xi=N_\eta=2^3$. In the extended spcace $p$, since $\lambda_{\mathrm{max}}^{+}(H_1)=4.8932$ and $\lambda_{\mathrm{max}}^{-}(H_1)=130.12$, we take the truncated interval in $p$ as $[-20.615,5.5872]$ with $N_p=2^{12}$ and $\Delta p=0.006369$. However, we choose the recovery point $p=0.59121> p^\diamond = \lambda^+_{\max}(H_1)T=0.58721$ and time step as $\Delta t =0.01$. 

Figure \ref{fig:I2ex3rhou} shows the exact, classical and quantum numerical solutions density $\rho$.

\begin{figure}[!htb]
    \centering
    \includegraphics[width=1.\linewidth]{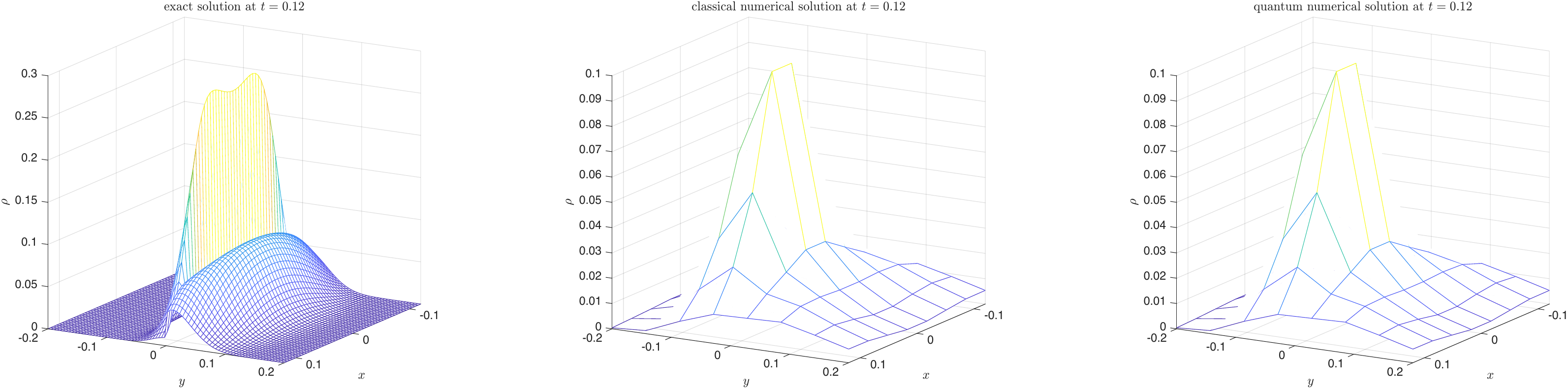}
    \caption{Example \ref{ex:I2ex3}, the density $\rho$ at $t = 0.12$. First column: the exact solution; second column: the classical numerical solution; third column: the Schr\"odingerization solution.}
    \label{fig:I2ex3rhou}
\end{figure}

Due to the overly coarse grid, the solution depicted in the figure is merely an approximation that roughly matches the true solution. When we aim to double the grid size in each dimension, the $p$-direction must be doubled twice as well, effectively increasing memory usage by a factor of 64. This results in insufficient memory. Nevertheless, in all of the above examples, we demonstrated that the Schr\"odingerization based Hamilton-preserving quantum algorithms do provide the correct numerical solutions, which agree well with the results computed by classical algorithms.  To demonstrate the real quantum advantages provided theoretically, one has to carry the computation in fault-tolerance quantum computers, which are not currently available.

\section{Conclusion} \label{sec:conclusion}

In this paper, we investigated the application of the Schr\"odingerization method to the Liouville equation in geometrical optics with interface conditions. By combining Schr\"odingerization with the Hamiltonian-preserving scheme in \cite{jinHamiltonianPreservingSchemeLiouville2006}, we developed quantum algorithms for both one- and two-dimensional cases. A key difficulty lies in the fact that partial transmission and reflection at the interface are classically treated through threshold-dependent ``if/else'' procedures, which makes it highly nontrivial to reformulate the scheme into a matrix form suitable for quantum simulation. To overcome this difficulty, we encoded the corresponding interface conditions into a matrix representation through suitable nonlinear functions prepared {\it a priori}. Based on this formulation, we analyzed the query complexity of the resulting quantum algorithms and proved that they achieve a polynomial quantum advantage in the precision parameter $\epsilon$ over their classical counterparts.

In the two-dimensional setting, the present work is restricted to the relatively simple case where the interface is aligned with the computational grid. An important direction for future research is to extend the Schr\"odingerization framework to more general interface geometries. Another natural extension is to consider Liouville equations with discontinuous potentials.

\section*{Acknowledgement}

SJ acknowledges the support of the NSFC grant No. 12341104, the Shanghai Pilot Program for Basic Research, the Science and Technology Commission of Shanghai Municipality (STCSM) grant no. 24LZ1401200, the Shanghai Jiao Tong University 2030 Initiative, and the Fundamental Research Funds for the Central Universities. SZ thanks Chuwen Ma at East China Normal University and Yue Yu at Xiangtan University for helpful and stimulating discussions.

\bibliographystyle{plain}
\bibliography{ref}

\end{document}